\documentclass[conference]{./IEEEtran}
\usepackage[T1]{fontenc}
\usepackage{cite}
\usepackage{hyperref}
\usepackage{graphicx}
\usepackage[left=2cm, right=2cm, top=1cm, bottom=2cm]{geometry}
\usepackage{amsmath}
\usepackage{amsfonts}
\usepackage{amsthm}
\usepackage[ruled,linesnumbered]{algorithm2e}
\usepackage{subcaption}
\usepackage{float}
\floatstyle{boxed} 
\restylefloat{figure}
\usepackage{amsfonts}
\usepackage{todonotes}
\usepackage[bb=boondox]{mathalfa}
\usepackage[lambda, primitives]{cryptocode}
\usepackage{dashbox}
\usepackage{tikz}
\usetikzlibrary{matrix,shapes,arrows,positioning,chains, calc}
\usepackage{mdframed}
\usepackage{pgfplots}
\usepackage{longtable}
\usepackage{multicol}
\pgfplotsset{width=10cm,compat=1.9}

\usepackage{enumitem}
\pgfplotsset{compat=1.18}
\usepackage{orcidlink}

\newtheorem{definition}{Definition}
\newtheorem{theorem}{Theorem}
\newtheorem{lemma}{Lemma}

\newcommand{\commit}[1]{\text{Com}\left(#1\right)}

\title{Verifiable Weighted Secret Sharing}

\begin{document}

\author{
    \IEEEauthorblockN{Kareem Shehata \orcidlink{0009-0004-7396-7236}}
    \IEEEauthorblockA{School of Computing \\
        National University of Singapore \\
        kareem@comp.nus.edu.sg
    }
    \thanks{This research is supported by Crystal Centre, National University of Singapore.}
    \and
    \IEEEauthorblockN{Han Fangqi \orcidlink{0009-0003-1895-8647}}
    \IEEEauthorblockA{School of Computing \\
        National University of Singapore \\
        fhan4@comp.nus.edu.sg
    }
    \and

    \IEEEauthorblockN{Sri AravindaKrishnan Thyagarajan \orcidlink{0000-0003-0114-7672}}
    \IEEEauthorblockA{School of Computer Science \\
        University of Sydney \\
        aravind.thyagarajan@sydney.edu.au
    }
}

\IEEEoverridecommandlockouts
\makeatletter\def\@IEEEpubidpullup{6.5\baselineskip}\makeatother
\IEEEpubid{\parbox{\columnwidth}{
    Crypto Valley Conference (CVC) 2025\\
    5-6 June 2025, Rotkreuz, Switzerland\\
    ISBN 978-1-XXXX-XXXX-1\\
    https://dx.doi.org/10.14722/cvc.2024.23xxx\\
    https://www.cryptovalleyconference.com/
}
\hspace{\columnsep}\makebox[\columnwidth]{}}

\maketitle

\pagenumbering{gobble}
\begin{abstract}
Traditionally, threshold secret sharing (TSS) schemes assume all
parties have equal weight, yet emerging systems like blockchains
reveal disparities in party trustworthiness, such as stake or
reputation. Weighted Secret Sharing (WSS) addresses this by assigning
varying weights to parties, ensuring security even if adversaries
control parties with total weight at most a threshold $t$. Current
WSS schemes assume honest dealers, resulting in security from only
honest-but-curious behaviour but not protection from malicious
adversaries for downstream applications. \emph{Verifiable} secret
sharing (VSS) is a well-known technique to address this, but existing
VSS schemes are either tailored to TSS, or
require additional trust assumptions. We propose the first efficient
verifiable WSS scheme that tolerates malicious dealers and is
compatible with the latest CRT-based WSS~\cite{crypto_w_weights}. Our
solution uses Bulletproofs for efficient verification and introduces
new privacy-preserving techniques for proving relations between
committed values, which may be of independent interest. Evaluation on
Ethereum show up to a $100\times$ improvement in communication
complexity compared to the current design and $20\times$ improvement
compared to unweighted VSS schemes.
\end{abstract}
\section{Introduction}
\label{sec:intro}

Secret sharing is a fundamental cryptographic building block. In a
secret-sharing scheme, a secret is distributed among a set of nodes
such that only authorized subsets can cooperate to recover it. Many
popular distributed protocols are built on secret sharing, including
multi-party computation, threshold encryption, threshold signatures,
and distributed randomness generation. The most well-known
secret-sharing scheme is Shamir's Secret Sharing (SSS)~\cite{sss},
which ensures that any set of corrupted parties up to a threshold $t$
cannot learn the secret. A key assumption implicit in SSS and other
threshold schemes is that all participating parties are of equal
importance, trust, or weight.

In many recent applications, the notion of equal weight fails to
capture real-world characteristics. For instance, in Proof-of-Stake
(PoS) blockchains~\cite{ethereum-pos}, each party is associated with
a stake that can vary significantly. Similarly, in Oracle
networks~\cite{oraclenets}, servers are treated according to their
respective reputation scores and an event is considered certified
only if it is endorsed by a set of servers with enough total
reputation score. Side-chains are another example where each party is
associated with a ``deposit'' of the main chain's currency, and the
amount of the deposit determines the party's voting power. These
scenarios, among many others, highlight the necessity for schemes in
which the adversary is modeled not by the number of individual nodes
it corrupts, but by the total weight of the compromised nodes.

The traditional solution is to ``virtualize'' parties, replacing a
party of weight $w$ with $w$ virtual parties of weight $1$, and using
an unweighted scheme such as SSS. This leads to communication and
computation costs of $O(w \cdot \lambda)$ per party, which grows
quickly in large systems with significant inequalities.

In \cite{crypto_w_weights}, Garg et al. propose an efficient weighted
secret sharing scheme based on the Chinese Remainder Theorem (CRT)
with costs growing linearly with a party's weight $w$, eliminating
the $\lambda$ factor. They describe a protocol with a trusted dealer
and explore applications such as multi-party computation (MPC),
threshold decryption, and threshold signatures. However, they only
consider an honest dealer, consequently their applications can only
tolerate semi-honest adversaries, i.e.\ parties are assumed to follow
the protocol specification, but may try to learn additional
information.

Verifiable secret sharing (VSS)~\cite{vss} allows shareholders to
verify the integrity of the deal, which is crucial when the dealer
cannot be fully trusted, or when using VSS as a subprotocol. As
stated in \cite{crypto_w_weights}, upgrading their CRT-based scheme
to VSS would enhance security to withstand malicious adversaries.
While VSS is well-studied for linear schemes such as SSS, the only
known CRT-based VSS~\cite{kaya_crt_ss} requires large unknown-order
groups, relies on the strong RSA assumption, and needs a trusted
setup. Since CRT-based secret sharing is non-linear, the techniques
used by SSS-based VSS schemes cannot be applied and achieving an
efficient CRT-based VSS without unknown-order groups has been
unclear.

Therefore, we ask the question, 
\begin{quote}
    {\em Can we construct an efficient CRT-based verifiable weighted secret sharing scheme without trusted setup, and assuming only a single group in which the discrete logarithm problem is hard?}
\end{quote}

In this paper, we answer the above question in the affirmative by enhancing the CRT-based weighted secret sharing scheme of \cite{crypto_w_weights} with verifiability of the shares in an efficient manner. We work with standard prime order elliptic curve groups (assuming the hardness of the discrete logarithm problem) which are widely used in real-world applications, particularly in the case of blockchain systems where weighted secret sharing has become prominent.

\subsection{Our Contributions}
\label{sec:contrib}

Our results in this work can be summarized as follows.

\smallskip\noindent\emph{Verifiable Weighted Ramp Secret Sharing.}
We construct the first efficient non-interactive CRT-based verifiable
weighted secret sharing scheme. Our scheme relies on prime
order groups in which the discrete logarithm problem is assumed to be
hard, such as an elliptic curve group, and does not require a trusted
setup. Our verifiability technique applies to \emph{any}
CRT-based secret sharing scheme, including the threshold construction
of \cite{asmuth_bloom} and the weighted-ramp\footnote{Ramp secret
sharing has two threshold parameters: $t$ and $T$. Parties with
aggregate weight above $T$ can reconstruct the correct secret, while
parties with aggregate weight below $t$ learn no information about
the secret.} construction of \cite{crypto_w_weights}. The scheme is
efficient, with communication cost scaling logarithmically in all
parameters.

\smallskip\noindent\emph{Proof-of-Mod.} The technical core of our construction is a
novel non-interactive zero-knowledge (NIZK) argument to prove that
two committed values $v, s$ satisfy $v = s \bmod p$, where $p$ is a
known prime. While this is straightforward for values within a small
subset of the commitment group (e.g.\ 64-bit values within a 256-bit
elliptic curve group), because of the wrap-around effect it was not
clear how to do this efficiently for arbitrary values until our work.
Further, we extend our approach to values that are much larger than
the group order, using prime-order
decomposition. The argument is efficient in that proof size is
logarithmic in the size of the inputs and may be of independent
interest.

\smallskip\noindent\emph{Practical evaluation.} We evaluate a case
study of Ethereum Staking and find that our scheme can reduce
bandwidth by up to a $20\times$ factor compared to unweighted
schemes. We also provide a proof-of-concept implementation and
confirm that proof sizes are logarithmic, under $2$ KiB for Ethereum.
Running times scale linearly with the number of parties and the size of the secret, but with large constants. With a more
optimized arithmetic circuit proof, our technique immediately becomes
practical.
\section{Overview}
\label{sec:overview}

In this section, we give an overview of our methods and the required
background while avoiding technical details. Please see
Sections~\ref{sec:proof-of-mod} and \ref{sec:crt-vss} for full
details.

\subsection{CRT-based Secret Sharing}
\label{sec:crt-ss}

Let us begin with a high-level description of Chinese Remainder
Theorem-based secret sharing~\cite{asmuth_bloom,mignotte}. We assume
a trusted dealer is given a secret $s_0 \in \mathbb{F}$, where
$\mathbb{F}$ is a field of size $p_0$. To share the secret among $n$
parties, we first choose $n$ numbers, $p_1, \ldots, p_n$, all
co-prime to each other and $p_0$. The choice of $p_i$ values depends
on the desired properties of the secret sharing scheme, which we
discuss later. For a parameter $L$, the dealer chooses uniformly
random $a \xleftarrow{\$} [L]$ and sets: \[ s = s_0 + a p_0 \le (L +
1) p_0 \]

We call $s$ the lifted secret. The dealer then calculates the share
for party $i$ as $s_i = s \mod p_i$, and sends each party their share
privately. Consider some set of parties $A \subseteq [n]$. Let $P_A =
\prod_{i \in A} p_i$. By the Chinese Remainder Theorem, the system of
equations $\left\{ s_i = s' \mod p_i \right\}_{i \in A}$ has a unique
solution modulo $P_A$, which can be solved to obtain $s'$. If $P_A >
(L + 1) p_0$, then $s' = s$, and we can find $s_0 = s' \mod p_0$. Let
$\mathcal{A}$ be the set of all access sets authorized to reconstruct
the secret. For reconstruction, we require that:
\[
  (L + 1) p_0 \le \min_{A \in \mathcal{A}} P_A = P_{MIN}
\]

Let $\bar{\mathcal{A}}$ be the set of all \emph{unauthorized} access
sets, define:
\[
  P_{MAX} = \max_{\bar{A} \in \bar{\mathcal{A}}} \prod_{i \in \bar{A}} p_i
\]

It can be shown that the statistical distance between the solution to
the CRT equations and uniformly random is at most
$P_{\bar{A}} / L$ (see~\cite{crypto_w_weights}), thus the security
error of our scheme is upper bounded by $P_{MAX} / L$. Our desired
property for both reconstruction and privacy of the system can thus
be succinctly stated as:

\[
  P_{MAX} << L < P_{MIN} / p_0 - 1
\]

\smallskip\noindent\textbf{Extension to Weighted Ramp Secret Sharing.}
\label{sec:wrss-overview}
The construction of a weighted secret-sharing scheme from the CRT
construction of the previous subsection is brilliant in its
simplicity. Let us assume that each party has an associated weight
$w_i$, and that there are two thresholds: a reconstruction threshold
$T$ and a privacy threshold $t$. A set of parties $A$ is authorized
if $\sum_{i \in A} w_i \ge T$, and unauthorized if $\sum_{i \in A}
w_i < t$. Notice that weights in between these two thresholds are
neither authorized nor unauthorized. We pick each prime $p_i$ such
that it has bit length $w_i$, i.e.\ $p_i$ is less than but close to
$2^{w_i}$. As a result, $P_{MAX} < 2^t$ and $P_{MIN} > 2^{T - O(1)}$.
Let $\lambda$ be our security parameter, and set $p_0$ to have bit
length $\lambda$. If we set $L = 2^{\lambda + t}$, then we achieve
security error $2^{-\lambda}$. Our system is correct if $T - t > 2
\lambda + O(1)$. If the gap is insufficient, we can amplify all of
the weights by a constant $c$. See~\cite{crypto_w_weights} for more a
more detailed analysis.

\subsection{CRT-based Verifiable Secret Sharing}
\label{sec:crt-vss-overview}

In addition to providing secret shares, in a VSS the dealer must also
prove to the participants that the deal was correctly executed
(see~\ref{sec:vss_defn} for a formal definition). For CRT-based
secret sharing schemes, the obvious approach is to provide
commitments to the secret and the shares, and then prove in
zero-knoweldge that the committed share values are congruent to the
original secret modulo the relevant prime. The participants verify
the proofs, and that their share is consistent with the given
commitments. If any of these checks fail they reject the deal.
Otherwise, they use CRT reconstruction to obtain the secret when
needed.

The main challenge is finding an appropriate commitment scheme and
proof technique. Using group-based commitments (e.g.\ Pedersen's) the
naïve approach leverages the homomorphic property: the prover
provides a commitment $K = \commit{k_i; r_i'}$ such that $s = s_i +
k_i p_i$, and then the verifier checks that $S = S_i \oplus (p_i
\otimes K_i)$, where $\oplus$ is the commitment group operation, and
$\otimes$ is the scalar operation. This technique works if $s$ is in
a small subset of the message space, but for arbitrary values, any
group-based commitment scheme faces a wraparound issue. Specifically,
if the commitment group has order $p_0$, then the equation above is
not over the integers but rather $\mathbb{Z}_{p_0}$. The dealer can
always find $k'$ such that $s = v + k' p \mod p_0$ for any values $v,
s$ since $\gcd(p, p_0) = 1$. As proven in
Appendix~\ref{sec:commit_wrap}, this problem affects all fixed-size
homomorphic commitments.

Unstructured commitment schemes (such as hash-based schemes) avoid
this problem, but there are currently no known efficient arithmetic
circuit proofs for such schemes. Thus, for arbitrary field values, it
is not clear how to efficiently construct a secure NIZK proof.
\footnote{One could always use Succinct Non-interactive Arguments
(SNARGs) for general circuits, but we believe a tailored NIZK proof
would easily outperform the SNARK approach both in prover and
verifier time.}


\subsection{Proof-of-Mod}
\label{sec:pom-overview}

The fundamental problem in creating a CRT-based VSS is proving that
committed values satisfy $v = s \bmod p$ for a known prime $p$
without revealing $v$ or $s$. Here we give an overview of our
"proof-of-mod" (PoM) protocol which achieves this goal, see
Section~\ref{sec:proof-of-mod} for full details.

\paragraph{Base PoM} To begin, suppose $s < p_0$ (i.e.\ $s$ is within
the order of the commitment group), and let $s = v + k p$, with $p_0
= q p + t$ where $0 \le t < p$ (i.e. $q$ is the quotient of $p_0 / p$
and $t$ is the remainder). In order to prove that $v = s \mod p$, it
suffices to show (1) that $0 \le v < p$ and (2) $0 \le k < q$. The
only other possibility is the last $t$ values that $s$ can take, in
which case we instead need to show (1) $0 \le v < t$ and (2) $0 \le k
\le q$. Our key observation is that the PoM boils down to a disjunction of two
sets of range proofs, each of which can be done within the group order:
\begin{align}
  \begin{split}
  \left[
    \left(0 \le v < p\right) \wedge
      \left(0 \le k < q\right)
  \right] \\
  \vee \left[
    \left(0 \le v < t\right) \wedge
      \left(0 \le k \le q\right)
  \right]
\end{split}
\label{eqn:range_dis}
\end{align}

\paragraph{Disjunction of Range Proofs} To understand how to prove a
disjunction of ranges, we first recap the range-proof technique
from Bulletproofs~\cite{bulletproofs}. To prove that $0 \le v_i < 2^n$,
where $v_i$ is the committed value, the prover shows that it knows the
$n$-bit binary decomposition of $v_i$, by showing that it knows values
$a_j, b_j$ that satisfy the following equations (see
Appendix~\ref{sec:range_eqns} for a proof):

\begin{align}
\label{eqn:range_ckt_1}
\sum_{j = 1}^{n} a_{i,j} \cdot 2^{j-1} &= v_i \\
\label{eqn:range_ckt_2}
a_{i,j} \cdot b_{i,j} &= 0 &1 \le j \le n \\
\label{eqn:range_ckt_3}
a_{i,j} - b_{i,j} - 1 &= c_{i,j} = 0 & 1 \le j \le n
\end{align}

\noindent Notice that if the range holds, then
(\ref{eqn:range_ckt_3}) is always zero. If the range does \emph{not}
hold, then the prover can always find $a_j, b_j$ that satisfy
(\ref{eqn:range_ckt_1}) and (\ref{eqn:range_ckt_2}) but not
(\ref{eqn:range_ckt_3}). Thus, to prove a disjunction of $0 < v_1 <
2^n$ or $0 < v_2 < 2^n$, we can form a polynomial for each $v_i$ that
evaluates to zero if $ \forall j, c_{i,j} = 0$, and non-zero with high
probability otherwise. Let $z$ be a random value chosen by the
verifier. This produces the following equations:
\begin{align}
\label{eqn:overview_c_poly}
\sum_{j = 1}^{n} c_{1,j} \cdot z^{j-1} &= c'_1, \ 
&\sum_{j = 1}^{n} c_{2,j} \cdot z^{j-1} = c'_2
\end{align}

If the disjunction is true then it must be the case that the product
of the two is zero, i.e., $c_1' \cdot c_2' = 0$. If neither range
holds, then with overwhelming probability the prover will not be able
to fulfil the proof for a random challenge $z$. We refer the reader to
Section~\ref{sec:rp-disj} for more details.

\paragraph{Prime-order decomposition}
For our CRT-based VSS to be useful, we require $s$ to be larger than
$p_0$. To that end, we use the technique of ``prime-order
decomposition'' which allows us to extend the base PoM above to
arbitrarily large values. Consider the case that $s = s_0 + a p_0$
for some $a \in \mathbb{N}$, where $s_0 < p_0$. We now wish to prove
that $v = s \bmod p$. Notice that we can now decompose the problem as
follows:
\begin{align*}
v &= (s_0 + a p_0) \bmod p \\
  &= (s_0 \bmod p) + (a \bmod p) (p_0 \bmod p)  
\end{align*}

If we assume that ${p}^2 + p < p_0$ then we can decompose the problem
into three smaller proofs for relations: (1) $(s_0' = s_0 \bmod p)$, (2) $(a' = a
\bmod p)$, and (3) $v = s_0' + a' t \bmod p$, where $t = p_0 \bmod p$.
The first two are simply applications of the base PoM, while the
third is an arithmetic circuit followed by PoM. We can similarly
extend this technique to any value polynomial in $p_0$.
For full details see Section~\ref{sec:full_mod}.

\paragraph{Efficiency}
\label{sec:eff-overview}

Proof sizes for the base PoM, extended PoM, and VSS are shown in
Table~\ref{tab:proof_efficiency}, assuming $p_0 < 2^\lambda$, $n$
parties, and lifted secrets less than ${p_0}^m$. The base proof-of-mod
is discussed in more detail in Section~\ref{sec:proof-of-mod}, while
the VSS proof is discussed in Section~\ref{sec:crt-vss}. The only
additional communication in the protocol is the $n+1$ commitments
$S_0, \ldots, S_n$, and $2n$ field elements $s_1, \ldots, s_n, r_1,
\ldots, r_n$, as required to fulfill the definition of a VSS.

\begin{table}[h]
\begin{center}
\caption{Proof sizes for PoM and VSS.}
\begin{tabular}{l|r|r}
& \multicolumn{1}{c|}{Group Elements} & \multicolumn{1}{c}{Field Elements} \\ \hline
Base PoM & $2 \lceil \log_2 (3\lambda + 1) \rceil + 8$      & $6$ \\ \hline
Ext. PoM & $2 \lceil \log_2 3m (3 \lambda + 1) \rceil + 8$  & $6$ \\ \hline
VSS      & $2 \lceil \log_2 3nm (3 \lambda + 1) \rceil + 8$ & $6$ 
\end{tabular}
\label{tab:proof_efficiency}
\end{center}
\end{table}

\subsection{Practical Parameters}
\label{sec:limits}

To use a single prime order group, we make the
assumption that the size of the group is large enough that all of the
sharing primes $p_i$ are smaller than the group order. As there is
also a limit on how small the sharing primes can be in practice, this
necessarily puts a limit on the dynamic range of weights that can be
supported. If our group order is $p_0 < 2^N$, then as we will see in
Section~\ref{sec:full_mod}, the maximum weight supported is less than
$N_{max} \le N/2$. On the other end, if the minimum practical weight
is $N_{min}$, then the dynamic range available is $N/2 - N_{min}$.

For example, if we assume a 256-bit group order (such as curve
Sec256k1), and a minimum prime length of 10 bits, then the dynamic
range\footnote{Dynamic range is the ratio of the maximum to the
minimum value.} available is 118. In other words, if the largest
party has $10\%$ of the total weight, then the smallest party that
can take part has weight $0.085\%$. We feel this is sufficient for
most practical purposes, including Ethereum (see
Section~\ref{sec:perf}), however if a larger dynamic range is needed,
then there are several solutions:

\begin{enumerate}

  \item \emph{Use a larger order group}: For example, Curve448 has order
  approximately 446 bits. With the same lower limit of 10 bits, this
  produces a dynamic range of 213. Using our example above, if the
  largest party has $10\%$ of the total weight, then the smallest
  party is now $0.047\%$ of the total weight.

  \item \emph{Weight pooling}: Similar to virtualization, it may be possible
  at a protocol level to have very small weight participants pool
  their weight together if they agree to act as one party. While this
  goes against the design of the protocol at a high level, for very
  small participants, this may be an acceptable solution if it
  results in considerable efficiency improvements.

  \item \emph{Virtualization}: On the opposite extreme, if there are very
  few very large participants, it may be possible to virtualize their
  shares. For example, if one party has $33\%$ of the total weight,
  then they could be split into three virtual parties, each with
  $11\%$ of the total weight. If all other parties have weight less
  than $11\%$, then only this party incurs the extra overhead. This
  produces significantly better fairness for smaller parties, while
  keeping the overhead for the large party quasi-logarithmic.

\end{enumerate}

\subsection{Publicly Verifiable Secret Sharing}
\label{sec:pvss-overview}

Notice that the above scheme requires a round of interaction between the participants to ensure a quorum of sufficient size has received valid shares and accepts the deal  As described in~\cite{pvss}, we can eliminate
this round of interaction to implement what is called \emph{Publicly
Verifiable Secret Sharing (PVSS)} if we use a Verifiable Encryption scheme,
such as Verifiable ElGamal (sometimes called ``ElGamal in the
exponent''). This allows the dealer to broadcast encryptions of the
share values, where the share of the $i$-th party is verifiably encrypted under the $i$-th party's public key. 
The dealer then attaches a NIZK proof to prove to (all) the participants that the encrypted
values match the committed values. 
 Thus the participants can verify
that the shares of all parties are correct without revealing them and
without any interaction. Our construction is amenable to the PVSS
scheme of~\cite{pvss} to eliminate interaction between participants
at the cost of a larger broadcast to all participants.

\subsection{Notation and Definitions}
\label{sec:notation}

We use $\lambda$ for the security parameter, which may be implicit in
a set of parameters, for example the parameters of the verifiable
secret sharing scheme. We use $\mathbb{F}$ to denote an arbitrary
field, and $\mathbb{Z}_p$ to denote the field of integers modulo a
prime $p$. Let $\text{negl}(\lambda)$ denote a negligible function,
that is for all polynomial $p(\lambda)$, there exists $\lambda_0$
such that $\text{negl}(\lambda) < 1/p(\lambda) \forall \lambda >
\lambda_0$. We denote by $[n]$ the set $\{1, \ldots, n\}$. We use the
$\vec{v}$ notation for vectors and uppercase bold font $\mathbf{M}$
for matrices.

\paragraph{Commitments} We use $\commit{x; r} \rightarrow c$ for a
commitment to $x$ using randomness $r$. Once an adversary provides
$c$ knowing $x$ and $r$, it should be hard to find another pair $(x',
r')$ such that $c = \commit{x'; r'}$ (binding). Without $x$ and $r$,
it should also be hard to determine which $x$ was committed to
(hiding). See Appendix~\ref{sec:commitments} for details and formal
definitions.

\paragraph{Arithmetic Circuit Proofs} We will make extensive use of
Arithmetic Circuit Proofs, a special case of Zero Knowledge Proofs
(ZKP)\footnote{A ZKP for computationally bounded adversaries is
technically an Argument, but we use the terms interchangeably.}. A
ZKP convinces a verifier that the prover knows a value $x$ satisfying
a relation $\mathcal{R}$ without revealing $x$. Here, the prover
shows that committed values satisfy an arithmetic circuit, defined by
multiplication gates $\vec{a} \circ \vec{b} = \vec{c}$ and
constraints on input and output values\footnote{Arithmetic circuits
of this form are also called Rank 1 Constraint Systems (R1CS).}. We
denote the prover and verifier algorithms as $\Pi_{CKT} = (P_{CKT},
V_{CKT})$, where $P_{CKT}((\vec{V}, \textbf{CKT}, A, B, C); (\vec{v},
\vec{a}, \vec{b}, \vec{c}, \vec{r}_v, (r_a, r_b, r_c)))$ is the
prover taking input commitment $\vec{V}$, circuit definition
$\textbf{CKT}$, wire commitments $A, B, C$, private inputs $\vec{v}$,
wire values $\vec{a}, \vec{b}, \vec{c}$, and randomness $\vec{r}_v,
r_a, r_b, r_c$. The verifier $V_{CKT}(\vec{V}, \textbf{CKT}, A, B, C)
\rightarrow \phi_{CKT}$ takes the same public inputs and returns
$\phi_{CKT} \in \{ \text{Accept}, \text{Reject} \}$. See
Appendix~\ref{sec:zk} and~\ref{sec:ckt_proofs} for formal definitions
and further details, and~\cite{bulletproofs} for an example
implementation.

\section{Proof-of-Mod}
\label{sec:proof-of-mod}

In this section we construct a NIZK proof that a committed value is
the result of taking another committed value modulo a known prime. We
start with the case that all values are within the order of the
discrete log group, and then show how to extent this technique to
arbitrary values. We build these proofs by constructing an arithmetic
circuit that is satisfied only if the proof condition holds. Since
such circuit proofs may be aggregated (see~\cite{bulletproofs}), we
can combine many such proofs into a much smaller proof than if they
were constructed individually.

\subsection{Initial Construction}
\label{sec:initial-construction}
\paragraph{Range Proof Circuits}
\label{sec:rp-ckts}
As we discussed in Section~\ref{sec:pom-overview}, we can use
Equations~\ref{eqn:range_ckt_1} to~\ref{eqn:range_ckt_3} to construct
an arithmetic circuit that is satisfied only if the input value $v_i$ is in
the range $[0, 2^n]$. We can tighten the range to $0 \le v_i < u_i$
by repeating this subcircuit, replacing $v_i$ with $u_i - v_i$. This
subcircuit is easily written in the form needed by circuit proofs
(see Appendix~\ref{sec:ckt_proofs}), requires $2n$ multiplications,
and $2n + 2$ constraints.

\paragraph{Disjunction of Ranges}
\label{sec:rp-disj}
We now construct a circuit for a disjunction of ranges, that is,
\emph{either} $0 \le v_1 < 2^n$ \emph{or.} $0 \le v_2 < 2^n$. Let us
assume that the prover chooses $a_{i,j}$ and $b_{i,j}$ such that
(\ref{eqn:range_ckt_1}) and ($\ref{eqn:range_ckt_2}$) hold for both
values, but (\ref{eqn:range_ckt_3}) holds for one of $v_1, v_2$
but not the other. In this case, we wish to prove that either \emph{all.}
$c_{1,j}$ values are zero, or \emph{all} $c_{2,j.}$ values are zero.
To this end we construct two polynomials:

\begin{align}
\label{eqn:c_poly}
\sum_{j = 1}^{n} c_{i,j} \cdot z^{j-1} &= c'_i
& i \in \{1, 2\}
\end{align}

After the prover has committed to the $a_{i,j}$ and $b_{i,j}$ values,
the verifier chooses a random $z \in \mathbb{F}$. If all of the
$c_{i,j}$ values are zero then $c'_i$ will be zero for any choice of
$z$. If at least one of the $c_{i,j}$ values is non-zero, then $c'_i$
will be non-zero with overwhelming probability. We can add one
more multiplication gate for $c_1' \cdot c_2' = 0$. If one of the
ranges hold, then this gate is satisfied. If neither range holds,
then with very high probability $c_1' \cdot c_2' \neq 0$. Notice that
we can replace the $n$ constraints in equation
(\ref{eqn:range_ckt_3}) with the single constraint in equation
(\ref{eqn:c_poly}) for each element in the disjunction by eliminating
the intermediate variables, thus reducing the number of constraints
from $2n$ to $2$.


It may seem that the prover cannot satisfy the circuit without first
knowing the choice of $z$, but there is a solution. The prover sets
$a_{i,1} = v_i$, $b_{i,1} = 0$, and all other $a_{i,j} = 0$, $b_{i,j}
= -1$. As a result, $c_{i,1} = v_i$ and all other $c_{i,j} = 0$
satisfies the circuit for any value $z$ chosen by the verifier.

\paragraph{Base Proof-of-Mod}
As described in Section~\ref{sec:pom-overview}, we first construct a
ZKA for a congruence relationship between two values
when all values are less than the order of the group used for the
commitments, which we call ``Base Proof-of-Mod''.

\begin{definition}[Base Proof-of-Mod]
  \label{def:base_pom}
  Given a commitment scheme $(\text{Setup}, \text{Com})$ over a
  message space $\mathcal{M} \subset \mathbb{N}$, a proof-of-mod is a
  zero-knowledge argument of knowledge for the relation
  $\mathcal{R}_{\text{mod}}$:

  \begin{multline}
    \mathcal{R}_{\text{mod}} ((V, S, p), (v, s, r_v, r_s)) := \\
    V = \commit{v, r_v} \wedge
    S = \commit{s, r_v}
    \wedge v = s \bmod p
  \end{multline}
  
  Where $V, S \in \mathcal{C}$, $p \in \mathbb{N}$, $v, s \in
  \mathcal{M}$, and $r_v, r_s \in \mathcal{R}$, and the modulo
  operation is over the natural numbers.

\end{definition}

As discussed in Section~\ref{sec:pom-overview}, the prover must show
that there exists a value $k$ such that $s = v + k p$ and that the
resulting value does not wrap around, i.e.\ that $v + kp < p_0$. This
results in the disjunction of ranges shown in
Equation~\ref{eqn:range_dis}, which we can use to construct a circuit
that is satisfied if and only if the proof condition holds. Let $p_0
= |\mathcal{M}|$ be the order of the commitment scheme, which we
assume to be prime, and $p_0 = q p + t$ where $q, t \in \mathbb{N}$
and $0 \le t < p$. Our circuit takes as input $v,s$ and is
constructed as follows:

\begin{enumerate}[nosep]
  \item Write constraint $s = v + k p$ for intermediate variable $k$.
  \item Write range subcircuits for $0 \le v < p$ and $0 \le k \le q$.
  \item Write disjunction subcircuit for $k < q$ or $v < t$.
\end{enumerate}

We show an optimized version of this circuit in
Appendix~\ref{sec:pom_eqns}, which uses $3n_1 + 3n_2 + 1$
multiplications and $5n_1 + 5n_2 + 9$ constraints, where $p <
2^{n_1}$ and $q < 2^{n_2}$. Let $\text{ModCkt}(p_0, p, z) \rightarrow
\mathbf{CKT}$ be a function that generates the circuit specification
using the supplied constants.

Let $\text{ModSolve}(v, s, p_0, p) \rightarrow \vec{a}, \vec{b},
\vec{c}$ be a function that produces the wire values that solve the
circuit produced by $\text{ModCkt}(p_0, p, z)$ for the given values.
An implementation of ModSolve is shown in
Appendix~\ref{sec:pom_eqns}. We can now write the base proof-of-mod
as shown in Figure~\ref{fig:base_pom}.

\DeclareExpandableDocumentCommand{\spancols}{O{l}mm}{%
\multicolumn{#2}{#1}{\ensuremath{#3}}
}

\begin{figure*}
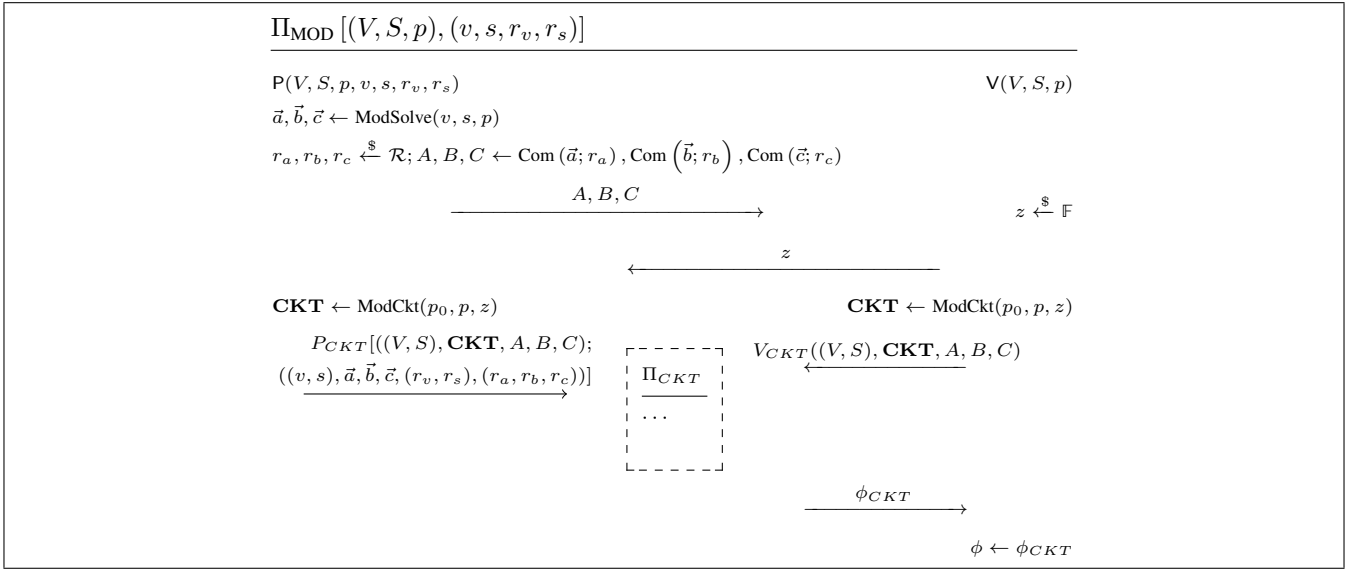

\procedureblock[codesize=\scriptsize]{$\Pi_{\text{MOD}}\left[(V, S, p), (v, s, r_v, r_s)\right]$}{
  \\
\< \prover(V, S, p, v, s, r_v, r_s) \< \< \> \verifier(V, S, p) \\
\< \spancols{5}{\vec{a}, \vec{b}, \vec{c} \leftarrow \text{ModSolve}(v, s, p)} \\
\< \spancols{5}{r_a, r_b, r_c \xleftarrow{\$} \mathcal{R};
    A, B, C \leftarrow \commit{\vec{a}; r_a}, \commit{\vec{b}; r_b}, \commit{\vec{c}; r_c}} \\
\< \sendmessagerightx[4cm]{5}{A, B, C} \> z \xleftarrow{\$} \mathbb{F} \\
\< \> \sendmessageleftx[4cm]{4}{z} \\
\< \mathbf{CKT} \leftarrow \text{ModCkt}(p_0, p, z)
  \< \> \spancols[r]{3}{\mathbf{CKT} \leftarrow \text{ModCkt}(p_0, p, z)} \\
  \< \sendmessageright*{
    P_{CKT}[((V, S), \mathbf{CKT}, A, B, C); \\
    ((v, s), \vec{a}, \vec{b}, \vec{c}, (r_v, r_s), (r_a, r_b, r_c))]}
    \< \dbox{\begin{subprocedure}
    \procedure{$\Pi_{CKT}$}{
      \> \cdots \\
    } \end{subprocedure}}
  \> \sendmessageleftx[2cm]{3}{V_{CKT}((V, S), \mathbf{CKT}, A, B, C)} \\
  \< \< \> \sendmessagerightx[2cm]{3}{\phi_{CKT}} \\
  \< \< \< \> \phi \gets \phi_{CKT}
}
\caption{Base Proof-of-Mod Protocol.}
\label{fig:base_pom}
\end{figure*}

\subsubsection{Security}
\label{sec:pom_sec}

Security rests on the security of the commitments and the circuit
proofs. We state our security theorem below, and provide a proof.

\begin{theorem}

If the commitment scheme $(\text{Setup}, \text{Com})$ used is
computationally binding and perfectly hiding, and the circuit proof
protocol $\Pi_{CKT}$ is perfectly complete, provides computational
soundness, and special honest verifier zero-knowledge, then
$\Pi_{\text{MOD}}$ is perfectly complete, provides computational
soundness, and special honest verifier zero-knowledge.

\end{theorem}

\begin{proof}

\emph{Completeness.} Given $v, s, r_v, r_s$ such that $v = s
\mod p$, the honest prover can always find a value of $k$ and all
intermediate values that satisfies the circuit for any value of $z$.
From the completeness of the circuit proof protocol, the verifier
will always accept.

\emph{Computational Soundness.} Consider a cheating prover $P^*$
that is given $v, s, r_v, r_s$ such that $v \neq s \mod p$, and both
the verifier and prover are given $V, S, p$. Assume $S = \commit{s,
r_s}, V = \commit{v; r_v}$, and $0 \le v < p$ since this is the most
advantageous condition for the prover. Consider the following cases.

\emph{Case 1.} $P^*$ finds either $v', r_v'$ such that $V =
\commit{v', r_v'}$ and $v' = s \mod p$, or $s', r_s'$ such that $S =
\commit{s', r_s'}$ and $v = s' \mod p$. In either case the prover has
broken the binding property of the commitment scheme, which by
assumption the probability of any PPT adversary finding such values
is negligible.

\emph{Case 2.} $P^*$ sends commitments $A \leftarrow
\commit{\vec{a}; r_a}, B \leftarrow \commit{\vec{b}; r_b}, C
\leftarrow \commit{\vec{c}; r_c}$, but then finds another set of
openings $\vec{a}', \vec{b}', \vec{c}', r_a', r_b', r_c'$ such that
$A = \commit{\vec{a}'; r_a'}, B = \commit{\vec{b}'; r_b'}, C =
\commit{\vec{c'}; r_c'}$ that satisfy the circuit. In this case, the
prover has broken the binding property of the commitment scheme. By
assumption the probability of any PPT adversary finding such values
is negligible.

\emph{Case 3.} $P^*$ sends commitments $A, B, C$ to values
$\vec{a}, \vec{b}, \vec{c}$ that do not satisfy the circuit, yet the
prover convinces the verifier to accept. In this case, the prover has
broken the soundness of $\Pi_{CKT}$. By assumption the probability of
any PPT adversary finding such values is negligible.

\emph{Case 4.} $P^*$ finds values that satisfy the disjunction
section of the circuit despite the disjunction being false. More
precisely, at least one of $a_{5,1} - b_{5,1} - 1 \neq 1$ and at
least one of $a_{6,1} - b_{6, 1} - 1 \neq 1$, yet one of $c_1', c_2'$
is zero (see Equation~\ref{eqn:c_poly} in Section~\ref{sec:rp-disj}).
By the fundamental theorem of algebra, a polynomial of degree $n$ has
at most $n$ roots. Thus, if $z$ is chosen uniformly at random, after
the polynomial is fixed, then the probability of this occurring is
$2n / |\mathbb{F}|$. Assuming $n$ grows at most linearly with the
security parameter $\lambda$ and $|\mathbb{F}|$ is exponential in the
same, then the probably of this occurring is negligible.

\emph{Case 5.} If none of the above cases occur, then the ranges
specified in Section~\ref{sec:pom-overview} are satisfied, and there
exists $k \in \mathbb{F}$ such that $s = v + k p < p_0$.

Since the three cases above cover all possibilities, from the law of
total probability it can be seen that the advantage of any PPT
adversary in the soundness game is negligible.

\emph{Special Honest Verifier Zero-Knowledge.} We show SHVZK by
building a simulator $\mathcal{S}$ that uses the simulator
$\mathcal{S}_{CKT}$ of $\Pi_{CKT}$ to simulate the protocol. First,
$\mathcal{S}$ chooses $z \in \mathbb{F}$ uniformly at random, and
constructs the circuit using $\text{ModCkt}(p_0, p, z)$. It then
calls $S_{CKT}$ to get the transcript of the circuit proof including
the commits $A, B, C$. It then outputs the transcript with the
addition of $z$ after $A, B, C$ are issued by the prover. By the
SHVZK property of $\Pi_{CKT}$, the transcripts produced by
$\mathcal{S}$ must have the same distribution as the transcripts
produced by the honest run of the protocol.
\end{proof}

\subsubsection{Efficiency}
\label{sec:pom_efficiency}

The circuit proof protocol $\Pi_{CKT}$ given in \cite{bulletproofs}
requires $2 \lceil \log_2 n \rceil + 8$ group elements and 5 field
elements, where $n$ in this case is the number of multiplication
gates. Our protocol adds only one field element, and uses $3 \lambda
+ 1$ multiplications, where $\lambda$ is the bit length of $p_0$.
Thus, in total our protocol requires $2 \lceil \log_2 (3\lambda + 1)
\rceil + 8$ group elements and $6$ field elements.

\subsection{Prime-Order Decomposition}
\label{sec:full_mod}

As discussed in Section~\ref{sec:pom-overview}, we can extend the base
proof-of-mod to arbitrary values by a technique we call ``prime-order
decomposition'', in which we decompose $s$ using $p_0$ as
a base. Consider that we can always write $s$ as:
\[ s = a_0 + a_1 {p_0} + a_2 {p_0}^2 + \dots + a_m {p_0}^m =
\sum_{i=0}^m a_i {p_0}^i \]

Where $a_i \in \mathbb{Z_{p_0}} \forall i \in [m]$. We can now define
an extended version of the proof-of-mod as follows.


\begin{definition}[Extended Proof-of-Mod]
  \label{def:ext_pom}
  Given a commitment scheme $(\text{Setup}, \text{Com})$ over a
  message space $\mathcal{M} \subset \mathbb{N}$, a proof-of-mod is a
  zero-knowledge argument of knowledge for the relation
  $\mathcal{R}_{\text{EMOD}}$:
  \begin{multline*}    
    \mathcal{R}_{\text{EMOD}} ((V, A_0, \ldots, A_m, p), \\
    (v, r_v, a_0, \ldots, a_m, r_0, \ldots, r_m)) := \\
    v = \sum_{i=0}^m a_i {p_0}^i \bmod p
    \wedge V = \commit{v, r_v} \wedge \\
    A_i = \commit{a_i, r_i} \forall i \in [m]
  \end{multline*}
  Where $V, A_0, \ldots, A_m \in \mathcal{C}$, $p \in \mathbb{N}$,
  $v, a_0, \ldots, a_m \in \mathcal{M}$, and $r_v, r_0, \ldots, r_m
  \in \mathcal{R}$, and the modulo operation is over the natural
  numbers.
\end{definition}

We will again take a circuit-building approach in order to prove
$\mathcal{R}_{\text{EMOD}}$. First, notice that we can decompose the
problem using Horner's method:
\begin{equation}
  \label{eqn:horner}
  v = a_0 + p_0 (a_1 + p_0 (a_2 + \dots + (a_{m-1} + p_0 a_m))) \mod p
\end{equation}

Let us define the intermediate steps in (\ref{eqn:horner}) as $v_i =
a_i + p_0 v_{i+1} = a_i + p_0 (\dots + p_0 a_m)$, where $v_m = a_m$.
Assuming $p_0 > p^2 + p$, then we can ensure that all values are within
$\mathbb{Z}_{p_0}$ by taking each step modulo $p$ as follows.
Let $a_i' = a_i \bmod p$ and $t = p_0 \bmod p$. Define:

\begin{align*}
  v'_i &= v_i \mod p \\
  &= a_i + p_0 \cdot v_{i+1} \mod p \\
  &= a'_i + t \cdot v'_{i+1} \mod p &\forall 0 \le i < m
\end{align*}

We can now apply the proof-of-mod circuit from
Section~\ref{sec:initial-construction} first to the $a_i'$ values,
and then to each of the $v'_i$, to build a circuit that proves the
correctness of $v$ with respect to $a_0, \ldots, a_m$. More
concretely, we construct the circuit as follows:

\begin{enumerate}[nosep]

  \item Write proof-of-mod subcircuits for $a_i' = a_i \bmod p$ for $i
  \in [0, m]$.

  \item Let $v'_m = a'_m$. For $j$ in $[0, m-1]$ write:

  \begin{enumerate}
    \item Constraint $v_j = a'_j + t \cdot v'_{j+1}$.
    \item Proof-of-mod subcircuit for $v'_j = v_j \bmod p$.
  \end{enumerate}

  \item Write constraint $v = v_0'$.

\end{enumerate}

In practice, the $a'_i, v_i, v_i'$ variables would be eliminated and
instead the bit decompositions used directly, but we present this
method for ease of understanding. If $p_0 > m p^2$ then instead of
$m$ PoM subcircuits for the $v_i'$ steps, we can directly write a
constraint for $v$ from the $a_i'$ values. Let $\text{EModCkt}(p_0,
p, z) \rightarrow \mathbf{CKT}$ be a function that generates the
circuit specification using the method described above, and let
$\text{EModSolve}(v, a_0, \ldots, a_m, p_0, p)$ be a function that
produces the wire values that solve $\mathbf{CKT}$. The resulting
protocol is identical to the previous protocol, replacing ModSolve
with EModSolve and ModCkt with EModCkt. Thus we elide the protocol
diagram and proof here. We denote the extended proof-of-mod protocol
proving $\mathcal{R}_{EMOD}$ as $\Pi_{EMOD}$.

\smallskip\noindent\textbf{Efficiency.}
\label{sec:full_pom_efficiency}
By constructing the proof as one large circuit we achieve
considerable efficiency in terms of proof size. As in the previous
subsection, each PoM requires $3 \lambda + 1$ multiplication gates,
where $p_0 < 2^\lambda$. The full construction uses $2m$ such PoM
circuits and no additional multiplications. Thus, the total
communication cost is $2 \lceil \log_2 m + \log_2 (3 \lambda + 1) +
\log_2 3 \rceil + 8$ group elements and $6$ field elements.

\paragraph{Non-Interactive Construction}
\label{sec:ni_pom}
The Fiat-Shamir heuristic~\cite{fiatshamir} allows us to convert any
public-coin interactive protocol into a non-interactive one with
computational soundness by replacing any randomness used by the verifier
with a hash of the transcript to that point. As is discussed in
detail in~\cite{bulletproofs}, the circuit proof protocol $\Pi_{CKT}$
can thus be rewritten in the form of a prover procedure that produces
a transcript, $\pi_{ckt}$, and a verifier procedure that takes
$\pi_{ckt}$ as input and outputs accept or reject.

Similarly, we can convert our proof-of-mod protocol into a non-interactive
proof simply by setting the verifier's challenge, $z$ to be the hash of
the commitments $A, B, C$ and the inputs to the protocol. We then use
the non-interactive version of the circuit proof as a subroutine to
get the full proof. We show the full details of the non-interactive
proof-of-mod protocol in Appendix~\ref{app:ni_pom_details}.

\section{CRT-Based VSS}
\label{sec:crt-vss}

Armed with the extended proof-of-mod we constructed in
Section~\ref{sec:proof-of-mod}, along with appropriate commitment and
arithmetic circuit proof schemes, we are now ready to flesh out the
CRT-based VSS scheme for which we gave a high level description in
Section~\ref{sec:crt-vss-overview}. We first construct a proof that
\emph{all} of the shares are correct by grouping all of the
proof-of-mods into a single circuit, and treating the $a_0, \ldots,
a_m$ values as intermediate variables that hold for all of the secret
shares. This way, separate commitments $A_1, \ldots, A_m$ are not
necessary. We can construct such a circuit taking $s_0, \ldots, s_n$
as inputs as follows:

\begin{enumerate}[nosep]
  \item Let $a_0, \ldots, a_m$ be intermediate variables
  \item Write constraint that $a_0 = s_0$.
  \item For $i$ in $1$ to $n$, write an extended proof-of-mod circuit
  for $s_i = \sum_{j=0}^m a_j p_0^j \mod p_i$.
\end{enumerate}

As in the Extended Proof-of-Mod construction, the intermediate
variables will be eliminated in an optimized circuit, but their
existence and consistency is still guaranteed. Let
$\text{VSSCKT}(p_0, \ldots, p_n, z) \rightarrow \mathbf{CKT}$ be a
function that constructs the circuit as described above using $z$ as
the challenge value for the proof-of-mod subcircuits, and let
$\text{VSSSolve}(s_0, \ldots, s_n, a_1, \ldots, a_m, p_0, \ldots,
p_n) \rightarrow \vec{a}, \vec{b}, \vec{c}$ be a function that
produces the wire values that solve $\mathbf{CKT}$ for any value of
$z$. We can now write our non-interactive VSS construction by
defining the Share and Reconstruct programs as shown in
Algorithms~\ref{alg:vss-share} and~\ref{alg:vss-reconstruct}
respectively.

This construction uses $n$ extended proof-of-mod circuits, each of
which uses $3m(2\lambda + 1)$ multiplication gates. Thus, the size of
$\pi_{ckt}$ is $2 \lceil \log_2 n + \log_2 m + \log_2 (3 \lambda + 1)
+ \log_2 3 \rceil + 8$ group elements and $6$ field elements,
including the commitments $A, B, C$. The only additional outputs are
the $n+1$ commitments $Y_0, \ldots, Y_n$ in $\pi$, and $2n$ field
elements $s_1, \ldots, s_n, r_1, \ldots, r_n$.

\SetAlgoSkip{}
\begin{algorithm}[h]
  \caption{VSS Share Procedure}
  \label{alg:vss-share}
  \KwIn{Secret $s_0 \in \mathbb{Z}_{p_0}$, primes $p_0, \ldots, p_n$}
  \KwOut{Shares $v_1, \ldots, v_n$, proof $\pi$}
  $a_1, \ldots, a_m \xleftarrow{\$} \mathbb{Z}_{p_0}$;
  $\vec{r} = r_1, \ldots, r_n \xleftarrow{\$} \mathcal{R}$;
  $Y_0 \gets \commit{s_0; r_0}$\;
  \For{$i = 1$ \KwTo $n$}{
    $s_i \gets s_0 + \sum_{j=0}^m a_j p_0^j \mod p_i$\;
    $Y_i \gets \commit{s_i; r_i}$\;
  }
  $\vec{a}, \vec{b}, \vec{c} \gets \text{VSSSolve}(s_0, \ldots, s_n, a_1, \ldots, a_m, p_0, \ldots, p_n)$\;
  $r_a, r_b, r_c \xleftarrow{\$} \mathcal{R};
  A, B, C \gets \commit{\vec{a}; r_a}, \commit{\vec{b}; r_b}, \commit{\vec{c}; r_c}$\;
  $z \gets H(p_0, \ldots, p_n, Y_0, \ldots, Y_n, A, B, C)$\;
  $\mathbf{CKT} \gets \text{VSSCkt}(p_0, \ldots, p_n, z)$\;
  $\pi_{ckt} \gets tr < \Pi_{CKT} > ((Y_0, \ldots, Y_n, A, B, C),$ $\mathbf{CKT},
  (s_0, \ldots, s_n), \vec{a}, \vec{b}, \vec{c}, \vec{r},
  r_a, r_b, r_c)$\;
  $\pi \gets (Y_0, \ldots, Y_n, A, B, C, \pi_{ckt})$\;
  \KwRet{$(s_0, r_0), \ldots, (s_n, r_n), \pi$}\;
\end{algorithm}

\begin{algorithm}[h]
  \caption{VSS Reconstruct Procedure}
  \label{alg:vss-reconstruct}
  \KwIn{Primes $p_0, \ldots, p_n$, shares $\{(i, v_i)\}_{i \in A}$, proof $\pi$}
  \KwOut{Secret $s'$ or $\bot$}
  $Y_0, \ldots, Y_n, A, B, C, \pi_{ckt} \gets \pi;
  s_i, r_i \gets v_i \ \forall i \in A$\;
  $z \gets H(p_0, \ldots, p_n, Y_0, \ldots, Y_n, A, B, C)$\;
  $\mathbf{CKT} \gets \text{VSSCkt}(p_0, \ldots, p_n, z)$\;
  \If{$V_{CKT}((Y_0, \ldots, Y_n), \mathbf{CKT}, A, B, C, \pi_{ckt})
        \neq \text{Accept}$ or $\exists i \in A\ | Y_i \neq \commit{s_i, r_i}$}{
    \KwRet{$\bot$}\;
  }
  $P_A \gets \prod_{i \in A} p_i$\;
  $q_i \leftarrow \prod_{j \in A \setminus \{i\}} p_j;
   q_i' \leftarrow {q_i}^{-1} \mod p_i\ \forall i \in A$\;
  $s' \leftarrow \sum_{i \in A} s_i q_i q_i' \mod P_A$\;
  \KwRet{$s' \mod p_0$}\;
\end{algorithm}

\paragraph{Security}
\label{sec:vss-security}

The security of our scheme rests on the security of the building
blocks. Secrecy follows from the hiding property of the commitment
scheme and the SHVZK property of the circuit proof. Correctness
follows from the perfect correctness of the PoM circuit. Finally,
commitment follows from the binding property of the commitment
scheme, and the soundness of the circuit proof. We state our final
theorem below and provide a proof.

\begin{theorem}
\label{thm:vss-deal-sec}

If the commitment scheme is computationally binding and perfectly
hiding, and the circuit proof used provides perfect completeness,
computational soundness, and SHVZK, then the NI-VSS shown in
Algorithms~\ref{alg:vss-share} and~\ref{alg:vss-reconstruct} is
secure using Definition~\ref{def:vss-sec}.

\end{theorem}

\begin{proof}
\emph{Secrecy.} Consider an arbitrary unauthorized access set
$\bar{A} \in \bar{\mathcal{A}}$, and any two secrets $s, s'$. By the
hiding property of the commitment scheme, the distribution of $Y_0,
\ldots, Y_n, A, B, C$ are independent and indistinguishable for any
two secrets $s, s'$. Similarly, by the SHVZK property of the circuit
proof, $\pi_{ckt}$ is indistinguishable for any two secrets $s, s'$.
The choices of $\{r_i\}_{i \in \bar{A}}$ are independent and
uniformly random, and thus also indistinguishable. Finally, we assume
that $p_0, \ldots, p_n$ are chosen such that for any $\bar{A} \in
\bar{\mathcal{A}}$, $P_{\bar{A}} \le P_{max} < p_0^{m+1}$. From
\cite{crypto_w_weights}, we can state the the statistical distance of
$\{s_i\}_{i \in \bar{A}}$ from uniform is at most $P_{max} /
p_0^{m+1}$. Thus, if $P_{max} << p_0^m$, then the statistical
distance between the output of $\text{Share}_{pp}$ for any $s, s'$ is
negligible.

\emph{Correctness.} Consider an arbitrary $A \in \mathcal{A}$.
First, note that the output of $\text{Share}_{pp}$ will always be
such that $s_i = s_0 + \sum_{i = 1}^{m-1} a_i p_0^m \mod p_i$, $Y_i =
\commit{s_i; r_i} \forall i \in [n]$. Next, from the perfect
correctness of the extended proof-of-mod circuit and the circuit
argument, the output values $A, B, C, \pi_{ckt}$ are such that if $z
= H(A, B, C)$ and $\mathbf{CKT} = \text{VSSCkt}(p_0, \ldots, p_n,
z)$, then $V_{CKT}((Y_0, \ldots, Y_n), \mathbf{CKT}, A, B, C,
\pi_{ckt})$ will always accept. Finally, by the Chinese Remainder
Theorem, if $p_0, \ldots, p_n$ are chosen such that $P_A > P_{min} >
p_0^{m+1}$, then the output of $\text{Reconstruct}_{pp}$ will always
be $s' = s_0$.

\emph{Commitment.} To prove that no adversary $D^*$ can in time
polynomial in the security parameter produce a proof and two sets of
shares $\pi, \{v_i\}_{i \in A}, \{v'_i\}_{i \in A'}$ where $A, A' \in
\mathcal{A}$ such that the two sets of shares reconstruct different
secrets, we consider the following cases.

\emph{Case 1.} Assume in this case that $D^*$ is able to commit
to shares $s_1, \ldots, s_n$ as $Y_1, \ldots, Y_n$ that are
consistent with the honest protocol, but then provides an opening for
at least one of the shares for a different value, $Y_i = \commit{s_i;
r_i} = \commit{s'_i; r'_i}$. By the computational binding of the
commitment scheme, the probability of a PPT $D^*$ finding such values
is negligible.

\emph{Case 2.} In this case, assume the dealer commits to $s_1,
\ldots, s_n$ as $Y_1, \ldots, Y_n$ such that the arithmetic circuit
is not satisfied, yet manages to produce a transcript $\pi_{ckt}$
such that $V_{CKT}((Y_0, \ldots, Y_n), \mathbf{CKT}, A, B, C,
\pi_{ckt})$ accepts. By the computational soundness of the circuit
proof, the probability of any PPT $D^*$ producing such a proof is
negligible.

\emph{Case 3.} In this case, assume that the commitments are
binding and that the circuit proof is sound, but that $D^*$ is able
to guess values of $\vec{a}, \vec{b}, \vec{c}$ that satisfy the
circuit for the value of $z$ obtained by the hash function. By the
random oracle model, we can model $H$ as a uniformly random function.
Thus the probability of $z$ satisfying the circuit for a given choice
of $\vec{a}, \vec{b}, \vec{c}$ is at most $n / |\mathbb{F}|$ from the
same argument given in Section~\ref{sec:pom_sec}. Thus, if $n$ grows
at most linearly in the security parameter, and $|\mathbb{F}|$ grows
exponentially, then each attempt has negligible probability of
success, and a polynomial number of such guesses by $D^*$ still
results in negligible probability of success.

\emph{Case 4.} In this case, we assume that the commitments are
binding, the circuit proof is sound, and the choice of $z$ is such
that no false roots are found. It must be the case that all of the
equations in the circuit are satisfied. Thus, there exists $a_0,
\ldots, a_m \in \mathbb{Z}_{p_0}$ such that for all $i \in [n]$, $s_i
= \sum_{j=0}^m a_j p_0^j \mod p_i$. As a result, the reconstruction
program will always output $s' = a_0$.

By the law of total probability, the probability of $D^*$ producing
inconsistent share values that pass the verification program is
negligible.
\end{proof}

\paragraph{Practical Parameters}
\label{sec:vss-limits}
As we discussed in Section~\ref{sec:wrss-overview}, in the WRSS
scheme, the weight of participant $P_i$ is determined by the length
of their associated prime $p_i$. In our construction, $p_0 > p_i^2 +
p_i$, and if $p_0 < 2^N$, then $p_i < 2^{\frac{N-1}{2}}$. If the minimum
practical value of $p_i$ is $N_{min}$, then the dynamic range
supported by the scheme is $N/2N_{min}$. For example, in Curve25519,
$p_0$ is approximately $2^{255}$, thus if $N_{min} = 10$ then $2^{10}
< p_i < 2^{127}$ and our dynamic range is $12.7$. Beyond this range
parties must again be virtualized, but notice that the virtualized
parties themselves have weight. Thus, the virtualization process is
not only more efficient, in that the number of virtual parties is
reduced by up to an order of magnitude, but also more fair, as the
weights are cumulative.
\section{Performance}
\label{sec:perf}

To evaluate the performance of our VSS, we take a case study of
Ethereum staking. We then provide a theoretical analysis of the
overhead of using a traditional VSS scheme with virtualization as
compared to our scheme. Finally, we implement our scheme to see the
operating parameters in practice.

\paragraph{Ethereum Staking}
\label{sec:eth}
Staking is the process used by proof-of-stake blockchains to validate
transactions. To participate, a user must deposit a certain amount of
Ether (the native currency of Ethereum) into a smart contract, which
is called their ``stake''. The user is then allowed to participate in
the consensus mechanism, and is rewarded for doing so. If the user
misbehaves, their stake is slashed, i.e.\ a portion of their deposit is
destroyed. A minimum stake of 32 Ether is required to
participate.~\cite{eth_staking}

The current system uses aggregate BLS signatures, which according to
one analysis~\cite{eth_sigs}, processes $~28,000$ signatures per
slot, and may rise to 1.8 million. Even with extremely efficient
aggregate signatures, this is a large overhead. The resulting
distribution of stakes is extremely wide, with the top two staking
pools (Lido and Coinbase) controling 45\% of the total stake. We show
the distribution of the rest of the system in
Figure~\ref{fig:eth_stake}.

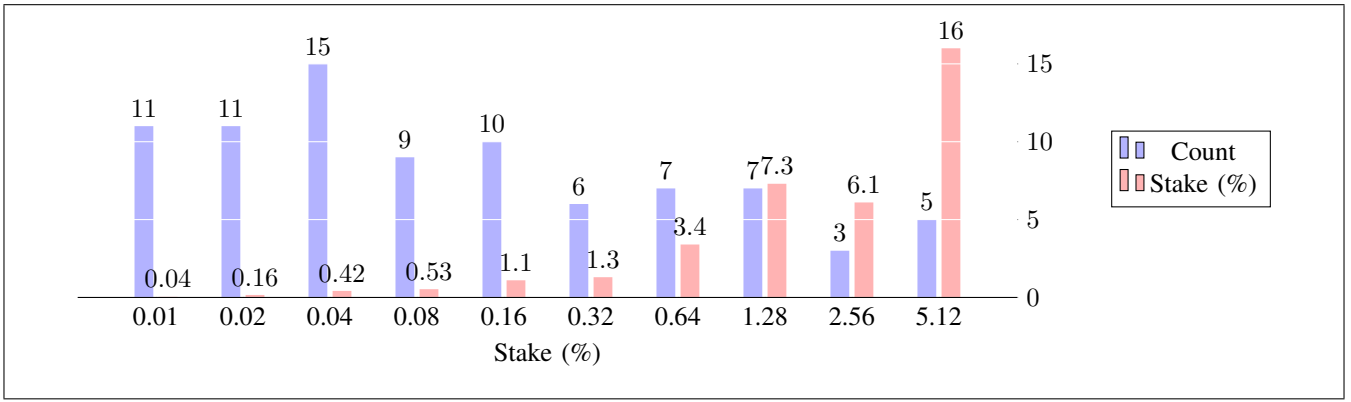
\begin{figure*}[t]
\begin{center}
\begin{tikzpicture}
    \centering
    \begin{axis}[
          ybar, axis on top,
          height=5cm, width=0.8 \textwidth,
          bar width=0.25cm,
          ymajorgrids, tick align=inside,
          major grid style={draw=white},
          enlarge y limits={value=.1,upper},
          ymin=0, ymax=15,
          axis x line*=bottom,
          axis y line*=right,
          y axis line style={opacity=0},
          tickwidth=0pt,
          enlarge x limits=true,
          legend style={
              at={(1.1,0.5)},
              anchor=west,
              legend columns=1,
              /tikz/every even column/.append style={column sep=0.5cm}
          },
          xlabel={Stake (\%)},
          symbolic x coords={
            0.01, 0.02, 0.04, 0.08, 0.16, 0.32, 0.64, 1.28, 2.56, 5.12},
         xtick=data,
         nodes near coords={
          \pgfmathprintnumber[fixed]{\pgfplotspointmeta}
         }
      ]
      \addplot [draw=none, fill=blue!30] coordinates {
        (0.01,11) (0.02,11) (0.04,15) (0.08,9) (0.16,10) (0.32,6) (0.64,7) (1.28,7) (2.56,3) (5.12,5)
        };
     \addplot [draw=none,fill=red!30] coordinates {
      (0.01,0.04) (0.02,0.16) (0.04,0.42) (0.08,0.53) (0.16,1.1) (0.32,1.3) (0.64,3.4) (1.28,7.3) (2.56,6.1) (5.12,16)
       };
      \legend{Count, Stake (\%)}
    \end{axis}
  \end{tikzpicture}
\caption{Distribution of Ethereum Stakes for pools other than Lido and Coinbase. Note that the x-axis is logarithmic.}
\label{fig:eth_stake}
\end{center}
\end{figure*}

\paragraph{Comparing Virtualization to WR-VSS}
\label{sec:vss_comparison}
Suppose Ethereum adopts a consensus scheme based on threshold
signatures, wherein validators receive shares of the signing key and
collaborate to jointly sign proposed blocks. This would require a
full distributed key generation protocol, but for the sake of
comparison consider only VSS overhead. As a baseline, consider the
Feldman VSS~\cite{feldmanvss}, a simple yet efficient VSS that uses
Shamir's Secret Sharing scheme applied to discrete-log
groups\footnote{More recent work provides better security, but to our
knowledge, it does not use less bandwidth.}. To share a secret among
$N$ (virtual) parties with threshold $t$, the scheme requires $N + t$
group elements to be broadcast, and one field element sent privately
to each (virtual) party.

From our analysis, such a scheme would require $4,110$ virtual
parties, and if $t = 2N/3$ then $6,850$ group elements must be
broadcast and $4,110$ field elements sent privately to each party.
This assumes a minimum stake of $0.02\%$ which excludes only $0.02\%$
of the current staking. This is not a large improvement over the
current design, and has problems with fairness%
\footnote{E.g.\ if a validator has $0.05\%$ of the total stake,
should they receive 2 or 3 votes?}.

By contrast, if Ethereum used our Weighted-Ramp VSS to generate the
signing key with a reconstruction threshold set to $2/3$ of the total
weight, and if the minimum stake of $0.02\%$ corresponds to a weight
of 10, then the total weight in the system is $41,125$. Using a curve
of order $\sim 2^{255}$, we set $m=108$ and need $365$ parties. As
can be seen in Table~\ref{tab:eth_bandwidth}, this produces more than
a $100\times$ improvement on the current design, and nearly
$20\times$ improvement in broadcast and $5\times$ improvement
in private bandwidth compared to the virtualized VSS.

Byte sizes assume Curve25519 for both VSS schemes, and BLS381 (48
byte group elements) for the current design. One last benefit to the
WRSS approach is that depending on the choice of signature scheme, it
may not be necessary to use a pairing-friendly curve (e.g.\ BLS12-381
as is currently in use) and instead, use a simpler curve (e.g.\
Curve25519 or sec256k1), which both simplifies the overall scheme and
reduces the bandwidth overhead.

\begin{table}[h]
\begin{center}
\caption{Comparison of bandwidth usage for current Ethereum signature broadcast, Feldman-VSS, and our WR-VSS with $0.02\%$
minimum stake and $t = 2/3 N$.
}
\begin{tabular}{l|r|r|r|r|r}
\multicolumn{1}{c|}{Design} & \multicolumn{3}{c|}{Broadcast}                          & \multicolumn{2}{c}{Private}       \\
 & $\mathbb{G}$ & $\mathbb{Z}_{p_0}$ & Total (B) & $\mathbb{Z}_{p_0}$ & Total (B) \\ \hline
Current    & $28,000$   &     & $1,344,000$ & &   \\ \hline
Feldman    & $6,850$    &     & $219,200$   &  $4,110$   & $131,520$   \\ \hline
WR VSS     & $389$      & $6$ & $12,640$    & $\sim 892$ & $28,528$ 
\end{tabular}
\label{tab:eth_bandwidth}
\end{center}
\end{table}

\paragraph{Implementation}
\label{sec:implementation}
To evaluate our scheme's practical performance, we implemented it in
Rust using the Bulletproofs library from \cite{dalekbp} for R1CS
proofs\footnote{R1CS, short for Rank-1 Constraint System, is another
term for arithmetic circuits.}. We tested on an Intel
Xeon Gold 6230 (2.10 GHz) running Ubuntu 20.04.2 LTS. Prover running
times for various configurations are shown in
Figure~\ref{fig:runtime}. As expected, time is linear in both $n$
(number of parties) and $m$ (lifting value). However, the constants
are large; even with this linear relationship, over 5 minutes is
needed for just 4 parties with $m=4$. We extrapolate from this that
applying our scheme to Ethereum would require hours of prover time.
Note that our implementation is not optimized, and the Bulletproofs
library is still in development, so significant improvements may be
possible.

The main advantage of our scheme is its small proof size. As shown in
Figure~\ref{fig:runtime}, proof size is logarithmic in both $m$ and
$n$, with small constants. We extrapolate that for the Ethereum
use cases, proof size should be under 2 KiB (excluding
commitments), aligning with theoretical expectations.

Thus, we have shown that our scheme achieves its primary goal of
proof size, though currently due to the poor performance of R1CS
libraries, the prover time is too high to be used in practice.

\begin{figure}[!h]
\begin{center}
\begin{subfigure}[b]{0.95\columnwidth}
    \pgfplotstableread[col sep=comma,]{data/runtime_m.csv}\runtimem
\begin{tikzpicture}
    \begin{axis}[
        width=0.9*\columnwidth,
        height=3.5cm,
        legend style={at={(0.02, 0.98)},anchor=north west},
        xlabel={$m$},
        ylabel={$t$ (s)},
        ]

        \addplot [scatter, only marks, mark=o, color=blue ] table [x={m}, y={1}]{\runtimem};
        \addplot[domain=1:4, samples=100, color=blue]{ -29.93567283 + x * 54.12952991 };

        \addplot [scatter, only marks, mark=x, color=red ] table [x={m}, y={2}]{\runtimem};
        \addplot[domain=1:4, samples=100, color=red]{ -29.93567283 + x * 2 * 54.12952991 };

        \addplot [scatter, only marks, mark=+, yellow ] table [x={m}, y={3}]{\runtimem};
        \addplot[domain=1:4, samples=100, color=yellow]{ -29.93567283 + x * 3 * 54.12952991 };

        \addplot [scatter, only marks, mark=*, green ] table [x={m}, y={4}]{\runtimem};
        \addplot[domain=1:4, samples=100, color=green]{ -29.93567283 + x * 4 * 54.12952991 };
    \end{axis}
\end{tikzpicture}
\subcaption{}
\label{fig:runtime_m}
\end{subfigure}
\begin{subfigure}[b]{0.95\columnwidth}
    \pgfplotstableread[col sep=comma,]{data/runtime_n.csv}\runtimen
\begin{tikzpicture}
    \begin{axis}[
        width=0.9*\columnwidth,
        height=3.5cm,
        legend style={at={(0.02, 0.98)},anchor=north west},
        xlabel={$n$},
        ylabel={$t$ (s)},
        ]

        \addplot [scatter, only marks, mark=o, color=blue ] table [x={n}, y={1}]{\runtimen};
        \addplot[domain=1:4, samples=100, color=blue]{ -29.93567283 + x * 54.12952991 };

        \addplot [scatter, only marks, mark=x, color=red ] table [x={n}, y={2}]{\runtimen};
        \addplot[domain=1:4, samples=100, color=red]{ -29.93567283 + x * 2 * 54.12952991 };

        \addplot [scatter, only marks, mark=+, yellow ] table [x={n}, y={3}]{\runtimen};
        \addplot[domain=1:4, samples=100, color=yellow]{ -29.93567283 + x * 3 * 54.12952991 };

        \addplot [scatter, only marks, mark=*, green ] table [x={n}, y={4}]{\runtimen};
        \addplot[domain=1:4, samples=100, color=green]{ -29.93567283 + x * 4 * 54.12952991 };

    \end{axis}
\end{tikzpicture}
\subcaption{}
\label{fig:runtime_n}
\end{subfigure}
\begin{subfigure}[b]{0.95\columnwidth}
\pgfplotstableread[col sep=comma,]{data/proof_size.csv}\proofsize
\begin{tikzpicture}
\begin{axis}[
    width=0.9*\columnwidth,
    height=3.5cm,
    legend style={at={(0.02, 0.98)},anchor=north west},
    xlabel={$m \times n$},
    ylabel={Proof Size (B)},
    ]

    \addplot [scatter, only marks, mark=o, color=blue ] table [x={nm}, y={proof}]{\proofsize};
    \addplot[domain=1:16, samples=100, color=blue]{ 1168.32346147 + 106.1508379 * ln(x) };
\end{axis}
\end{tikzpicture}
\subcaption{}
\label{fig:proof_size}
\end{subfigure}
\caption{VSS Measurements: (a) prover time vs $m$ (size of the lifted secret),
(b) prover time vs $n$ (number of parties), (c) proof size.}
\label{fig:runtime}
\end{center}
\end{figure}
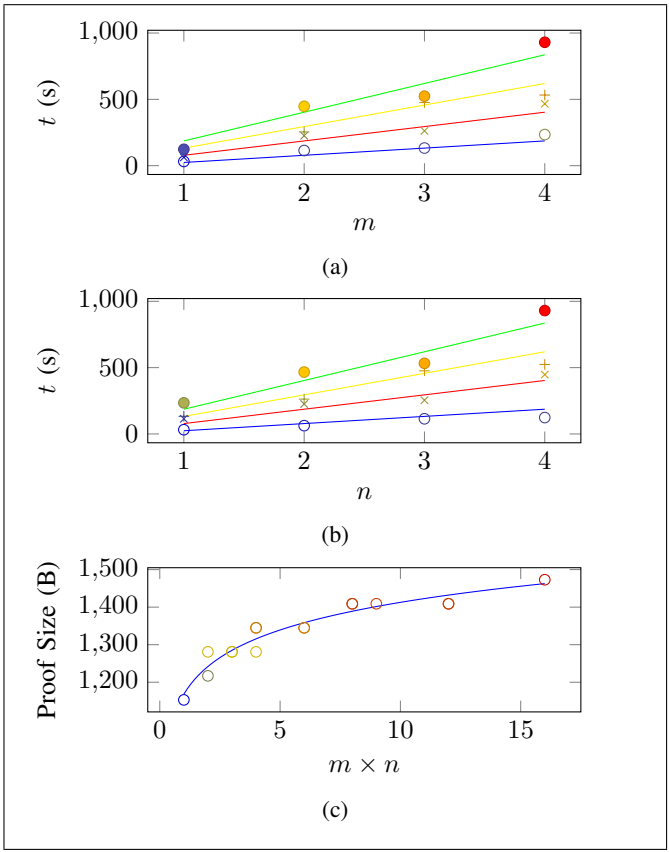

\section{Conclusions and Future Work}
\label{sec:conclusions}

In this work we have shown how to construct a Verifiable
Weighted-Ramp Secret Sharing scheme. Our scheme is based on the CRT
and is the first to provide verifiability for a weighted secret
sharing scheme without a trusted setup, unknown order groups, or the
strong RSA assumption. Along the way we developed a novel technique
for proving congruence relationships between committed values in zero
knowledge that may be of independent interest. We have shown that our
scheme is efficient, with communication costs scaling logarithmically
in all parameters, and running times linear in all parameters.

At present, the scheme is implemented using the Bulletproofs
arithmetic circuit proof system, which introduces large constants in
the prover time. Future work may optimize this by using more
efficient arithmetic circuit proof systems, at which point the prover
time may be more practical for large systems. It may also be possible
to trade off some of the prover time for increased proof size.
Finally, future work may consider the case that the a party's prime
value is greater than the group order, which would require a
different proof-of-mod system, in order to enable larger dynamic
ranges.

\bibliographystyle{IEEEtranS}
\bibliography{paper}

\newpage
\appendix
\section{Definitions}
\label{sec:definitions}

\subsection{Commitments}
\label{sec:commitments}

We will use the formal definition of commitments as below. In
general, Pedersen commitments are used in practice in order to use
the same discrete log group for commitments and the rest of the
system.

\begin{definition}[Commitment]

  A non-interactive commitment scheme is a pair of randomized
  polynomial time algorithms $(\text{Setup}, \text{Com})$:

  \begin{itemize}[leftmargin=*]

    \item $\textit{Setup}(1^\lambda) \rightarrow \text{pp}$: Takes as
    input a security parameter $\lambda$ and outputs the parameters
    for the commitment scheme, $pp$.

    \item $\textit{Com}_{pp}(x; r) \rightarrow c$ takes as input a message $x
    \in \mathcal{M}_{pp}$ and randomness $r \in \mathcal{R}_{pp}$ and
    outputs a commitment $c \in \mathcal{C}_{pp}$.

  \end{itemize}

\end{definition}

For ease of notation we will often omit the subscript $pp$ when the
setup step is clear. In general, it is expected that $r$ is drawn
uniformly at random from the randomness space $\mathcal{R}_{pp}$ when
generating a fresh commitment.

\begin{definition}[Binding Commitment]

  A commitment scheme is said to be binding if for all PPT
  adversaries $\mathcal{A}$ there exists a negligible function $\mu$
  such that:

  \[
    \Pr\left[
      \begin{array}{l}
        \commit{x_0; r_0} \\
        = \commit{x_1; r_1} \\
        \wedge x_0 \neq x_1
      \end{array}
    \middle\vert
    \begin{array}{l}
      pp \leftarrow \text{Setup}(1^\lambda), \\
      x_0, r_0, x_1, r_1 \leftarrow \mathcal{A}(pp)
    \end{array}
    \right]
    \le \mu(\lambda)
  \]

  Where probability is taken over all random coins of Setup and
  $\mathcal{A}$. If $\mu(\lambda) = 0$ then we say that the scheme is
  perfectly binding.

\end{definition}

\begin{definition}[Hiding Commitments]

  A commitment scheme is said to be hiding if for all PPT adversaries
  $\mathcal{A}_1, \mathcal{A}_2$ there exists a negligible function
  $\mu$ such that:

\[
  \left|
  \Pr\left[
    b = b'
  \middle\vert
  \begin{array}{l}
    pp \leftarrow \text{Setup}(1^\lambda), 
    b \xleftarrow{\$} \{0,1\}, \\
    r \xleftarrow{\$} \mathcal{R}_{pp}, 
    x_0, x_1 \leftarrow \mathcal{A}_1(pp), \\
    c \leftarrow \commit{x_b; r}, 
    b' \leftarrow \mathcal{A}_2
  \end{array}
  \right]
  - \frac{1}{2}
  \right|
  \le \mu(\lambda)
\]
  
Where probability is taken over $b, r$, and all random coins of Setup
and $\mathcal{A}$. If $\mu(\lambda) = 0$ then we say that the scheme
is perfectly hiding.

\end{definition}

\paragraph{Vector Commitments} A vector commitment scheme is simply
a commitment scheme where the message space $\mathcal{M}$ is a vector
space, such as $\mathbb{Z}_p^n$. The definitions for homomorphic,
binding, and hiding commitments apply equally to vector commitments.

\subsection{Commitment Wraparound}
\label{sec:commit_wrap}

\begin{definition}[Homomorphic Commitments]

  A homomorphic commitment scheme is a non-interactive commitment
  scheme such that $\mathcal{M}_{pp}, \mathcal{R}_{pp},
  \mathcal{C}_{pp}$ are all abelian groups, and for all $x, y \in
  \mathcal{M}_{pp}$ and $r_x, r_y \in \mathcal{R}_{pp}$:

  \[
   \commit{x; r_x} \otimes \commit{y; r_y} = \commit{x \oplus y; r_x \oplus r_y}
  \]

  Where $\otimes$ denotes the group operation in $\mathcal{C}_{pp}$ and
  $\oplus$ denotes the group operation in $\mathcal{M}_{pp}$ or
  $\mathcal{R}_{pp}$ respectively.

\end{definition}

\paragraph{Wraparound} Notice that by the properties of abelian
groups, homomorphic commitments over a finite commitment space such
that $|\mathcal{C}| = p$ must necessarily wrap around, such that
$\commit{x + p; r} = \commit{x; r}$. This is obvious for Pedersen
commitments, but to see that this is true in general, consider a
generator $g \in \mathcal{C}$ such that $\commit{\mathbb{1}; 0} = g$,
for some interpretation of $\mathbb{1}$. For any integer value $p$ we
can write $y = \mathbb{1} \cdot p$ where scalar multiplication in
this case is interpretted as the group operation repeated $p$ times.
We can thus write:

\begin{align*}
  \commit{x + p; r} &= \commit{x; r} \otimes \commit{\mathbb{1} \cdot p; 0} \\
  &= \commit{x; r} \otimes g^p
\end{align*}

By Fermat's theorem, $g^p$ is the identity element of the group.

As a concrete example, many proofs including Bulletproofs, use
Pedersen Commitments, which for a group $\mathbb{G}$ with generators
$g, h$ (generated by Setup) of prime order $p$ in which the discrete
logarithm problem is hard, are defined as:

\[ \commit{x; r} := g^x h^r \]

Pedersen commitments are homomorphic, \emph{computationally} binding,
and \emph{perfectly} hiding. Note that the discrete logarithm of $h$
with respect to $g$ and vice versa must be unknown to the prover.
This notion can be extended to a vector commitment by having setup
generate a vector of generators such that:

\[
  \commit{\vec{x}; r} := \vec{g}^{\vec{x}} h^r
  = h^r \prod_{i=1}^n g_i^{x_i}
\]

\subsection{Zero-Knowledge Arguments}
\label{sec:zk}

Intuitively, a zero-knowledge proof of knowledge is an interactive
protocol between a prover P and a verifier V such that on input $x$
the prover proves to the verifier that it knows a witness $w$ that
satisfies a relation $\mathcal{R}(x, w)$, without revealing anything
else about $w$. If the proof is sound for only computationally
bounded provers, we call it a zero-knowledge argument (ZKA). If the
verifier keeps no secrets, then we say the proof (or argument) is
public coin. In this paper we consider only public-coin
zero-knowledge arguments, and use the following formal definition.

\begin{definition}[Zero-Knowledge Arguments]
  \label{def:zk}

  Let $\mathcal{R}$ be an NP relation. A zero-knowledge argument
  (ZKA) $\Pi$ for $\mathcal{R}$ is an interactive protocol between a
  prover $P$ and a verifier $V$ consisting of $k$ prover messages and
  $k-1$ verifier messages. The prover is defined by a family of
  randomized algorithms $P = \{P_i\}_{i \in [k]}$, while the verifier
  is defined by a predicate $\phi$ such that:

  \begin{itemize}[leftmargin=*]

    \item The prover $P$ on the $i$-th round takes as input $x$ and
    the witness $w$ along with all previous challenge messages $c_j$
    for $j \in [i-1]$ and a random input $r_i$ and outputs a message
    $m_i$. Written formally:

    \[ P_i(x, w, \{c_j\}_{j \in [i-1]}; r_i) \rightarrow m_i,
    \forall i \in [k] \]

    \item The verifier $V$ on the $i$-th round samples a random
    challenge $c_i \xleftarrow{\$} \{0,1\}^\kappa,\  \forall i \in
    [k-1]$ and sends it to the prover.

    \item After receiving all $k$ messages from the prover, the
    verifier outputs $b \leftarrow \phi(x, \{m_i\}_{i \in [k]},
    \{c_i\}_{i \in [k-1]})$.

  \end{itemize}

  We denote by $\langle P', V' \rangle$ the random variable
  corresponding to the output $b$ of $V'$ after interacting with $P'$
  in the execution of $\Pi$, and by $\text{TR}(P', V') = (m_1, c_1,
  \ldots, c_{k-1}, m_k)$ the transcript of the execution of $\Pi$ in
  which $V'$ interacts with $P'$. We define the following security
  notions:

  \emph{Completeness}: We say that $\Pi$ has perfect
  completeness if for all $(x, w) \in \mathcal{R}$:
  \[ \Pr[\langle P(x, w), V(x) \rangle = 1] = 1 \]

  \emph{Computational Soundness}: We say that $\Pi$ has
  computational soundness if for all $(x, w) \notin \mathcal{R}$
  and PPT provers $P^*$:
  \[ \Pr[\langle P^*(x, w), V(x) \rangle = 1] \leq negl(\lambda) \]

  \emph{Special Honest Verifier Zero-Knowledge (SHVZK)}: We
  say that $\Pi$ has SHVZK if there exists a PPT simulator
  $\mathcal{S}$ such that for any $(x, w) \in \mathcal{R}$ the
  simulator $\mathcal{S}$ on input $x$ outputs a transcript that is
  indistinguishable from the transcript of the honest prover $P$
  run with input $x$ and witness $w$. Written formally:
  \[
    \{ \text{TR}(P(x, w), V(x)) \} \approx_c \{ \mathcal{S}(x) \}
  \]

\end{definition}

\subsection{Arithmetic Circuit Proofs}
\label{sec:ckt_proofs}

Informally, an arithmetic circuit proof is a zero-knowledge argument
that a set of committed values satisfy a given arithmetic circuit. As
we discussed in Section~\ref{sec:overview}, we will make extensive
use of an efficient proof of circuit satisfiability such as the
construction given in \cite{zkp_ckts} and refined in
\cite{bulletproofs}. We assume that the commitments and the circuit
operate on some field $\mathbb{F}$, which in practical terms will
generally be $\mathbb{Z}_p$ for some prime $p$.

\sloppy Consider a set of $n$ values $(v_1, \ldots, v_n)$, with commitment
vector $\vec{V} = (V_1, \ldots, V_n) \in \mathcal{C}^n, V_i =
\commit{v_i; r_{v_i}}$. We can define an arithmetic circuit by
starting with $m$ multiplication gates. We label the inputs to the
$i$-th gate as $a_i$ and $b_i$, and the output as $c_i$. We can
represent the circuit using the following two sets of equations:
\begin{align}
  a_i \cdot b_i &= c_i &1 \le i \le m \\
  \sum_{i = 1}^m w_{j, a, i} a_i 
  &+ \sum_{i = 1}^m w_{j, b, i} b_i \nonumber \\
  + \sum_{i = 1}^m w_{j, c, i} c_i 
  &= \sum_{i = 1}^n w_{j, v, i} v_i + k_j &1 \le j \le q
\end{align}
Where $w_{j, x, i}$ are weights determined by the circuit. We can
write this more succinctly by grouping the values into vectors and
matrices in the obvious way.
\begin{align}
  \label{eqn:ckt_sat_1}
  \vec{a} \circ \vec{b} &= \vec{c} \\
  \label{eqn:ckt_sat_2}
  \mathbf{W}_a \cdot \vec{a}
  + \mathbf{W}_b \cdot \vec{b}
  + \mathbf{W}_c \cdot \vec{c}
  &= \mathbf{W}_v \cdot \vec{v} + \vec{k}
\end{align}

We say that a circuit given by $\mathbf{W}_a, \mathbf{W}_b,
\mathbf{W}_c, \mathbf{W}_v$ is satisfied by values $\vec{v}$ if there
exists $\vec{a}, \vec{b}, \vec{c}$ such that equations
\ref{eqn:ckt_sat_1} and \ref{eqn:ckt_sat_2} hold. This leads to a
natural definition of a circuit proof.

\begin{definition}[Arithmetic Circuit Proof]
  \label{def:ckt_proof}

  Given a commitment scheme $(\text{Setup}, \text{Com})$ an
  arithmetic circuit proof is a zero-knowledge argument of knowledge
  for the relation $\mathcal{R}_{\text{ckt}}$:
\begin{multline}
  \mathcal{R}_{\text{ckt}} (
    (\vec{V} \in \mathcal{C}^n,
      \mathbf{W}_a, \mathbf{W}_b,
        \mathbf{W}_c \in \mathbb{F}^{q \times m},
      \mathbf{W}_v \in \mathbb{F}^{q \times n}, \\
      \vec{k} \in \mathbb{F}^q),
    (\vec{v} \in \mathbb{F}^n,
      \vec{a}, \vec{b}, \vec{c} \in \mathbb{F}^m,
      \vec{r}_v \in \mathcal{R}^n)) \\
  := \\
    V_i = \commit{v_i, r_{v_i}} \forall i \in [1,n] \\
    \wedge \vec{a} \circ \vec{b} = \vec{c} \\
  \wedge \mathbf{W}_a \cdot \vec{a}
  + \mathbf{W}_b \cdot \vec{b}
  + \mathbf{W}_c \cdot \vec{c}
  = \mathbf{W}_v \cdot \vec{v} + \vec{k}
\end{multline}

\end{definition}

\begin{theorem}

  There exists an arithmetic circuit proof protocol
  $\Pi_{\text{ckt}}$ with perfect completeness, computational
  soundness, and perfect special honest verifier zero-knowledge for
  $\mathcal{R}_{\text{ckt}}$. Moreover, $\Pi_{\text{ckt}}$ has has
  communication complexity $O(\log m)$.

\end{theorem}

We refer the reader to \cite{bulletproofs} for the security proof.

\paragraph{Notation} We will generally write $\textbf{CKT}$ in place
of $(\mathbf{W}_a, \mathbf{W}_b, \mathbf{W}_c, \mathbf{W}_v,
\vec{k})$ when the context is clear, implying that the conditions
above are satisfied. For example, invoking the circuit proof protocol
we will write $\Pi_{\text{ckt}} \left[ (\vec{V}, \textbf{CKT});
(\vec{v}, \vec{a}, \vec{b}, \vec{c}, \vec{r}_v )\right]$.

\subsection{Verifiable Secret Sharing}
\label{sec:vss_defn}

Intuitively, a verifiable secret sharing scheme (VSS) is a secret
sharing scheme with added functionality that allows the participants
to verify both the dealer actions and the shares received by other
participants. Let $pp$ be the parameters of the scheme, including the
number of shares $n$ and any parameters needed to determine the set
of authorized access sets as $\mathcal{A}_{pp} = \{ A \subseteq [n] |
A \text{ authorized to reconstruct the secret} \}$, and similarly
$\bar{\mathcal{A}}_{pp}$ the set of all access sets explicitly
unauthorized to reconstruct the secret, where $\mathcal{A}_{pp} \cap
\bar{\mathcal{A}}_{pp} = \emptyset$. We will consider only
non-interactive VSS, which is composed of two programs:

\begin{enumerate}[leftmargin=*]
  
  \item $\text{Share}_{pp}(s) \rightarrow ((v_1, \ldots, v_n), \pi)$,
  the sharing program takes as input a secret value $s$ and produces
  $n$ shares of the secret and a proof $\pi$. The proof $\pi$ is
  broadcast to all parties, while $v_i$ is given only to party $i$.

  \item $\text{Reconstruct}_{pp}(\{(i, v_i \in \mathcal{K})\}_{i \in
  A}, \pi) \rightarrow s' \in \mathcal{K} | \bot$, the reconstruction
  programs take a set of shares and reconstructs the secret as $s'$
  or $\bot \notin \mathcal{K}$.

\end{enumerate}

In practice there will generally also be a verification program for
accepting the shares. As this is not needed for the security of our
VSS, and should be quite clear from the construction, we omit the
formal definition.

\begin{definition}
\label{def:vss-sec}

We say that a VSS is secure if for any legal parameters $pp$ that
contains a security parameter $\lambda$, the following properties
hold:

\paragraph{Secrecy} For any unauthorized set $\bar{A} \in
\bar{\mathcal{A}}_{pp}$, and any two secrets $s, s'$ , the
statistical distance between the distribution of the shares produced
by $\text{Share}_{pp}(s)$ and $\text{Share}_{pp}(s')$ is negligible.
That is:

\begin{align}
\begin{split}
  SD (
  \left\{ \pi, \{v_i\}_{i \in \bar{A}} |
    (v_1, \ldots, v_n, \pi) 
  \leftarrow \text{Share}_{pp}(s) \right\}, \\
  \left\{ \pi', \{v'_i\}_{i \in \bar{A}} |
    (v'_1, \ldots, v'_n, \pi') \leftarrow \text{Share}_{pp}(s') \right\}
    ) \\
  \le \text{negl}(\lambda)
\end{split}
\end{align}

\paragraph{Correctness} For any authorized set $A \in \mathcal{A}$,
and any secret $s$, if $(v_1, \ldots, v_n, \pi) \leftarrow
\text{Share}_{pp}(s)$ then:
\[
  \text{Reconstruct}_{pp}(\{v_i\}_{i \in A}, \pi) = s
\]

\paragraph{Commitment} For any PPT adversary $D^*$ taking input $pp$
and outputing $\pi, \{v_i\}_{i \in A}, \{v'_i\}_{i \in A'}$ the
probability that the two sets of shares reconstruct different secrets
is negligible. That is:

\begin{multline}
  \Pr[\bot \neq
  \text{Reconstruct}_{pp}(\{v_i\}_{i \in A}, \pi) \\ \neq
  \text{Reconstruct}_{pp}(\{v_i\}_{i \in A'}, \pi) \neq
  \bot] \le \text{negl}(\lambda)
\end{multline}

\end{definition}

\section{Base Proof-of-Mod Details}
\label{app:base_pom}

\subsection{Proof of Range Equations}
\label{sec:range_eqns}

\begin{lemma}
\label{lem:bit_decomp}

  Let $v$ be an integer, and $\mathbb{F}$ a field of size at least 2.
  For $n > 0 \in \mathbb{N}$, $v \in [0, 2^n]$ if and only If there
  exists vectors $\vec{a}, \vec{b} \in \mathbb{F}^n$ such that the
  following equations hold:

  \begin{align}
    \label{eqn:bit_decomp_1}
    \sum_{i = 1}^{n} a_i 2^{i-1} &= v \\
    \label{eqn:bit_decomp_2}
    a_i \cdot b_i &= 0 &1 \le i \le n \\
    \label{eqn:bit_decomp_3}
    a_i - b_i - 1 &= 0 &1 \le i \le n
  \end{align}

\end{lemma}

\begin{proof}

Rearranging (\ref{eqn:bit_decomp_3}) and
substituting into (\ref{eqn:bit_decomp_2}) gives:
\[
  a_i (a_i - 1) = 0
\]

Thus either $a_i = 0$ or $a_i = 1$. Applying this to
(\ref{eqn:bit_decomp_1}):

\[
  0 \le \sum_{i = 1}^{n} a_i 2^{i-1} = v \le \sum_{i = 1}^{n} 2^{i-1} < 2^n
\]
\end{proof}

\subsection{Base Proof-of-Mod Circuit Equations}
\label{sec:pom_eqns}

Inputs: $v, s$. \\
Constants: $p_0, n_1, n_2 \in \mathbb{N}, p, q \in \mathbb{Z}_{p_0}, t \in \mathbb{Z}_{p}$, such that:

\begin{align*}
  p_0 = pq + t \\
  2^{n_1} \ge p \\
  2^{n_2} \ge q
\end{align*}

\begin{figure*}[h!]
\begin{align*}
\sum_{i = 1}^{n_2} p 2^{i-1} a_{i} &= s - v &(pk = s - v) \\
\sum_{i = 1}^{n_2} 2^{i-1} a_{i}
  + \sum_{i = 1}^{n_2} 2^{i-1} a_{n_2 + i}
  &= q &(k \le q) \\
\sum_{i = 1}^{n_1} 2^{i-1} a_{2 n_2 + i} &= v &(0 \le v) \\
\sum_{i = 1}^{n_1} 2^{i-1} a_{n_1 + 2 n_2 + i} &= p - v - 1 &(v \le p - 1) \\
\sum_{i = 1}^{n_2} 2^{i-1} a_{i}
  + \sum_{i = 1}^{n_2} 2^{i-1} a_{2 n_1 + 2 n_2 + i}
  &= q  - 1 &(k \le q - 1) \\
\sum_{i = 1}^{n_1} 2^{i-1} a_{2 n_1 + 3 n_2 + i} &= t - v - 1 &(v \le t - v) \\
a_i - b_i - 1 &= 0 &1 \le i \le 2 n_1 + 2 n_2 \\
\sum_{i = 1}^{n_2} z^{i-1} (a_{2 n_1 + 2 n_2 + i} - b_{2 n_1 + 2 n_2 + i} - 1)
- a_{3 n_1 + 3 n_2 + 1} &= 0 \\
\sum_{i = 1}^{n_1} z^{i - 1} (a_{2 n_1 + 3 n_2 + i} - b_{2 n_1 + 3 n_2 + i} - 1)
- b_{3 n_1 + 3 n_2 + 1} &= 0 \\
c_i &= 0 &1 \le i \le 3n_1 + 3n_2 + 1 \\
\end{align*}
\caption{Base Proof-of-Mod Circuit Equations.}
\label{fig:pom_eqns}
\end{figure*}

Equations for the base proof-of-mod circuit are shown in
Figure~\ref{fig:pom_eqns}, written n the form required for
$\Pi_{ckt}$ in Section~\ref{sec:ckt_proofs}. The equations in
brackets on the right are meant as an intuitive explanation.
Equations eliminate $k$ as an intermediate value, and use $a_1,
\ldots, a_{n_2}$ instead. This construction uses $3n_1 + 3n_2 + 1$
multiplication gates and $5n_1 + 5n_2 + 9$ constraints.

\subsection{ModSolve Function}
\label{sec:modsolve}

Let $\text{Decomp}(v, n) \rightarrow \vec{a}, \vec{b}$ be a a
function that decomposes $v$ into its binary representation in $n$
bits, such that $\langle \vec{a}, \vec{2}^n \rangle = v$, $\vec{a}
\circ \vec{b} = \vec{0}^n$, and $\vec{a} - \vec{b} - \vec{1} =
\vec{0}^n$, if such a representation is possible, or $a_1 = v, a_2 =
\ldots = a_n = 0, b_1 = 0, b_2 = \ldots = b_n = -1$ otherwise.
Algorithm~\ref{alg:modsolve} shows the ModSolve function that calculates
the wire values that solve the equations in the previous section.

\begin{algorithm}[h]
  \caption{ModSolve Function for Base Proof-of-Mod.}
  \label{alg:modsolve}
  \KwIn{$v \in \mathbb{Z}_p$, $s \in \mathbb{Z}_{p_0}$, $p_0, p \in \mathbb{N}$, such that $v = s \mod p$.}
  \KwOut{$\vec{a}, \vec{b}, \vec{c}$ that solve the equations in Figure~\ref{fig:pom_eqns}.}
  $n_1 \leftarrow \lceil \log_2 p \rceil$\;
  $n_2 \leftarrow \lceil \log_2 \lfloor p_0 / p \rfloor \rceil $\;
  $k \leftarrow (s - v) \cdot p^{-1} \mod p_0 $\;
  $\vec{a}_1, \vec{b}_1 \leftarrow \text{Decomp} (k) $\;
  $\vec{a}_2, \vec{b}_2 \leftarrow \text{Decomp} (q - k) $\;
  $\vec{a}_3, \vec{b}_3 \leftarrow \text{Decomp} (v) $\;
  $\vec{a}_4, \vec{b}_4 \leftarrow \text{Decomp} (p - v - 1) $\;
  $\vec{a}_5, \vec{b}_5 \leftarrow \text{Decomp} (q - k - 1) $\;
  $\vec{a}_6, \vec{b}_6 \leftarrow \text{Decomp} (t - v - 1) $\;
  $a_7 \leftarrow a_{5,1} - b_{5,1} - 1 $\;
  $b_7 \leftarrow a_{6,1} - b_{6,1} - 1 $\;
  $\vec{a} \leftarrow \vec{a}_1 || \vec{a}_6 || a_7 $\;
  $\vec{b} \leftarrow \vec{b}_1 || \vec{b}_6 || b_7 $\;
  $\vec{c} \leftarrow \vec{0}^{3n_1 + 3n_2 + 1} $\;
  \KwRet{$\vec{a}, \vec{b}, \vec{c}$}\;
\end{algorithm}

\section{Non-Interactive Proof-of-Mod Implementation}
\label{app:ni_pom_details}

First,
we assume a family of secure has functions $\mathcal{H}$, from which
we select $H \in \mathcal{H}$ for use in the protocol below. Next we
assume that the Fiat-Shamir heuristic is applied to the circuit proof
protocol using the same family of hash functions, and that the first
message from the prover to the verifier includes vector commitments
to the wires of the circuit as $A, B, C$. Let $tr<\Pi_{CKT}>
(\vec{V}, \mathbf{CKT}, \vec{v}, \vec{a}, \vec{b}, \vec{c},
\vec{r_v}, A, B, C, r_a, r_b, r_c) \rightarrow \pi$ be a
polynomial-time function that produces an accepting transcript $\pi$
for the resulting non-interactive proof that does not include $A, B,
C$, and let $V_{CKT}(\vec{V}, CKT, A, B, C, \pi_{ckt})$ be a function
that verifies transcript $\pi_{ckt}$ for $\Pi_{CKT}$. We can then
construct a non-interactive proof-of-mod protocol as follows using
the procedures shown in Algorithm~\ref{alg:ni_prover} and
Algorithm~\ref{alg:ni_verifier} for the prover and verifier
respectively.

\begin{algorithm*}[h]
  \caption{Non-interactive Prover for Proof-of-Mod.}
  \label{alg:ni_prover}
  \KwIn{
    $p_0, p \in \mathbb{N}$,
    $a_0, \ldots, a_m \in \mathbb{Z}_{p_0}$,
    $v = \sum_{i = 0}^{m} a_i {p_0}^i \mod p \in \mathbb{Z}_p$,
    $r_v, r_0, \ldots, r_m \in \mathcal{R}$,
    $V = \commit{v; r_v}$,
    $A_0, \ldots, A_m = \commit{a_0; r_0}, \ldots, \commit{a_m; r_m}$.
    }
  \KwOut{Proof $\pi$.}
$\vec{a}, \vec{b}, \vec{c} \leftarrow \text{EModSolve}(v, a_0, \ldots, a_m, p_0, p)$\;
$r_a, r_b, r_c \xleftarrow{\$} \mathcal{R}$\;
$A, B, C \leftarrow \commit{\vec{a}; r_a}, \commit{\vec{b}; r_b}, \commit{\vec{c}; r_c}$\;
$z \leftarrow H(A, B, C)$\;
$\mathbf{CKT} \leftarrow \text{EModCkt}(p_0, p, z)$\;
$\pi_{ckt} \leftarrow tr< \Pi_{CKT} >((V, A_0, \ldots, A_m), \mathbf{CKT}, 
(v, a_0, \ldots, a_m), \vec{a}, \vec{b}, \vec{c}, (r_v, r_0, \ldots,
  r_m), A, B, C, r_a, r_b, r_c)$\;
  \KwRet{$\pi = (A, B, C, \pi_{ckt})$}\;
\end{algorithm*}

\begin{algorithm*}[h]
  \caption{Non-interactive Verifier for Proof-of-Mod.}
  \label{alg:ni_verifier}
  \KwIn{
    $p_0, p \in \mathbb{N}$,
    Commitments $V, A_0, \ldots, A_m$,
    Proof $\pi$.
    }
  \KwOut{Accept or Reject.}
  $A, B, C, \pi_{ckt} \leftarrow \text{Parse}(\pi)$\;
  $z \leftarrow H(A, B, C)$\;
  $\mathbf{CKT} \leftarrow \text{EModCkt}(p_0, p, z)$\;
  \uIf{$V_{CKT}((V, A_0, \ldots, A_m), \mathbf{CKT}, A, B, C, \pi_{ckt}) = \text{Accept}$}{
    \KwRet{Accept}\;
  } \Else{
    \KwRet{Reject}\;
  }
\end{algorithm*}

\onecolumn
This page is needed for a bug in long table. Please delete this page.
\pagebreak

\subsection{Ethereum Staking Data}
\label{sec:eth_data}

\begin{center}
\begin{longtable*}{l|r|r|r|r|r}
\multicolumn{1}{c|}{Entity} & \multicolumn{1}{|c|}{ETH Staked} & \multicolumn{1}{|c|}{Marketshare} & \multicolumn{1}{|c|}{Virtual Shares} & \multicolumn{1}{|c|}{WRSS Bits} & \multicolumn{1}{|c}{WRSS Shares} \\ \hline \hline
Lido & 9,370,076 & 31.7588 & 1588 & 1611 & 14 \\ \hline
Coinbase & 4,329,345 & 14.6738 & 734 & 749 & 7 \\ \hline
Binance & 1,162,816 & 3.9412 & 197 & 208 & 2 \\ \hline
Kiln & 989,472 & 3.3537 & 168 & 178 & 2 \\ \hline
Figment & 937,632 & 3.1780 & 159 & 169 & 2 \\ \hline
Rocket Pool & 853,934 & 2.8943 & 145 & 155 & 2 \\ \hline
Kraken & 825,601 & 2.7983 & 140 & 150 & 2 \\ \hline
Staked.us & 666,492 & 2.2590 & 113 & 123 & 1 \\ \hline
OKX & 587,649 & 1.9918 & 100 & 109 & 1 \\ \hline
Bitcoin Suisse & 541,942 & 1.8368 & 92 & 102 & 1 \\ \hline
Upbit & 374,912 & 1.2707 & 64 & 73 & 1 \\ \hline
stakefish & 374,752 & 1.2702 & 64 & 73 & 1 \\ \hline
Mantle & 338,048 & 1.1458 & 57 & 67 & 1 \\ \hline
DARMA Capital & 326,112 & 1.1053 & 55 & 65 & 1 \\ \hline
Blockdaemon & 263,712 & 0.8938 & 45 & 54 & 1 \\ \hline
Frax Finance & 231,488 & 0.7846 & 39 & 49 & 1 \\ \hline
P2P.org & 230,816 & 0.7823 & 39 & 48 & 1 \\ \hline
Swell & 182,688 & 0.6192 & 31 & 40 & 1 \\ \hline
ether.fi & 167,844 & 0.5689 & 28 & 38 & 1 \\ \hline
Daniel Wang & 151,648 & 0.5140 & 26 & 35 & 1 \\ \hline
CoinSpot & 137,664 & 0.4666 & 23 & 33 & 1 \\ \hline
Diva (Pre-launch) & 125,443 & 0.4252 & 21 & 30 & 1 \\ \hline
Octant & 100,000 & 0.3389 & 17 & 26 & 1 \\ \hline
Stader & 82,565 & 0.2798 & 14 & 23 & 1 \\ \hline
Stakewise & 76,672 & 0.2599 & 13 & 22 & 1 \\ \hline
MyEtherWallet & 63,235 & 0.2143 & 11 & 20 & 1 \\ \hline
XHash & 60,800 & 0.2061 & 10 & 19 & 1 \\ \hline
imToken & 60,128 & 0.2038 & 10 & 19 & 1 \\ \hline
Bitstamp & 52,032 & 0.1764 & 9 & 18 & 1 \\ \hline
Revolut & 45,152 & 0.1530 & 8 & 17 & 1 \\ \hline
Gate.io & 41,536 & 0.1408 & 7 & 16 & 1 \\ \hline
StakeHound & 37,504 & 0.1271 & 6 & 15 & 1 \\ \hline
Liquid Collective & 35,520 & 0.1204 & 6 & 15 & 1 \\ \hline
Poloniex & 31,072 & 0.1053 & 5 & 14 & 1 \\ \hline
RockX & 29,216 & 0.0990 & 5 & 14 & 1 \\ \hline
KuCoin & 26,464 & 0.0897 & 4 & 14 & 1 \\ \hline
BlockFi & 26,112 & 0.0885 & 4 & 13 & 1 \\ \hline
Bitfinex & 24,387 & 0.0827 & 4 & 13 & 1 \\ \hline
Stkr (Ankr) & 24,104 & 0.0817 & 4 & 13 & 1 \\ \hline
Everstake & 22,784 & 0.0772 & 4 & 13 & 1 \\ \hline
Harbour & 20,288 & 0.0688 & 3 & 12 & 1 \\ \hline
arthapala.eth & 18,816 & 0.0638 & 3 & 12 & 1 \\ \hline
Consensys & 18,240 & 0.0618 & 3 & 12 & 1 \\ \hline
Bake & 17,408 & 0.0590 & 3 & 12 & 1 \\ \hline
WEX Exchange & 16,000 & 0.0542 & 3 & 12 & 1 \\ \hline
conurtrol.eth & 15,840 & 0.0537 & 3 & 12 & 1 \\ \hline
StakeWise & 13,824 & 0.0469 & 2 & 11 & 1 \\ \hline
Node DAO & 12,288 & 0.0416 & 2 & 11 & 1 \\ \hline
Bitpie & 11,456 & 0.0388 & 2 & 11 & 1 \\ \hline
HTX & 10,368 & 0.0351 & 2 & 11 & 1 \\ \hline
Mercado Bitcoin & 10,368 & 0.0351 & 2 & 11 & 1 \\ \hline
Celsius & 8,993 & 0.0305 & 2 & 11 & 1 \\ \hline
bitshameddesk.eth & 8,640 & 0.0293 & 1 & 10 & 1 \\ \hline
BTC-e & 8,416 & 0.0285 & 1 & 10 & 1 \\ \hline
honoraryape.eth & 8,000 & 0.0271 & 1 & 10 & 1 \\ \hline
Taylor Gerring & 8,000 & 0.0271 & 1 & 10 & 1 \\ \hline
was.eth & 7,680 & 0.0260 & 1 & 10 & 1 \\ \hline
Paxos & 7,680 & 0.0260 & 1 & 10 & 1 \\ \hline
Swell (Pre-launch) & 7,188 & 0.0244 & 1 & 10 & 1 \\ \hline
CoinDCX & 7,136 & 0.0242 & 1 & 10 & 1 \\ \hline
Vitalik Buterin & 6,976 & 0.0236 & 1 & 10 & 1 \\ \hline
Tranchess & 6,976 & 0.0236 & 1 & 10 & 1 \\ \hline
cryptostake.com & 6,592 & 0.0223 & 1 & 10 & 1 \\ \hline
guccilorian.eth & 5,088 & 0.0172 &  &  &  \\ \hline
for.eth & 5,056 & 0.0171 &  &  &  \\ \hline
EPotter & 4,640 & 0.0157 &  &  &  \\ \hline
Ebunker & 4,640 & 0.0157 &  &  &  \\ \hline
Sigma Prime Team & 4,608 & 0.0156 &  &  &  \\ \hline
Nimbus Team & 4,608 & 0.0156 &  &  &  \\ \hline
Teku Team & 4,608 & 0.0156 &  &  &  \\ \hline
Prysm Team & 4,608 & 0.0156 &  &  &  \\ \hline
SharedStake & 3,584 & 0.0121 &  &  &  \\ \hline
StaFi & 3,456 & 0.0117 &  &  &  \\ \hline
Redacted Pirex & 3,378 & 0.0114 &  &  &  \\ \hline
MintDice.com & 2,688 & 0.0091 &  &  &  \\ \hline
Bitget & 2,688 & 0.0091 &  &  &  \\ \hline
Bifrost & 1,632 & 0.0055 &  &  &  \\ \hline
DxPool & 992 & 0.0034 &  &  &  \\ \hline
Blox Staking & 992 & 0.0034 &  &  &  \\ \hline
staked.finance & 960 & 0.0033 &  &  &  \\ \hline
Flipside & 512 & 0.0017 &  &  &  \\ \hline
neukind.com & 480 & 0.0016 &  &  &  \\ \hline
Uphold & 448 & 0.0015 &  &  &  \\ \hline
ClayStack & 128 & 0.0004 &  &  &  \\ \hline
SenseiNode & 96 & 0.0003 &  &  &  \\ \hline
Other Solo Stakers & 151,143 & 0.5123 &  &  &  \\ \hline
Unidentified & 5,026,342 & 17.0362 &  &  & 
\end{longtable*}
\end{center}

\end{document}